\documentclass[11pt]{article}

\usepackage{fullpage}
\usepackage{amsfonts, amssymb, amsmath, amsthm}
\usepackage{latexsym}
\usepackage{url}
\usepackage{color}
\definecolor{DarkGray}{rgb}{0.1,0.1,0.5}
\usepackage[colorlinks=true,breaklinks, linkcolor= DarkGray,citecolor= DarkGray,urlcolor= DarkGray]{hyperref}	
\usepackage{graphicx}
\usepackage{subfigure}
\usepackage{cancel}
\usepackage{multirow}
\usepackage{enumitem}  
\usepackage{caption}	

\usepackage{import}
\usepackage{ltxtable}

\usepackage{calc}	

\def\place #1#2#3{\mspace{#2}\makebox[0pt]{\raisebox{#3}{#1}}\mspace{-#2}}	

\newcommand{\bra}[1]{{\langle#1|}}
\newcommand{\ket}[1]{{|#1\rangle}}
\newcommand{\braket}[2]{{\langle#1|#2\rangle}}
\newcommand{\ketbra}[2]{{\ket{#1}\!\bra{#2}}}
\newcommand{\abs}[1]{{\lvert #1\rvert}}	
\newcommand{\norm}[1]{{\| #1 \|}}

\newcommand{\beq}{\begin{equation}}
\newcommand{\eeq}{\end{equation}}

\def\tensor {\otimes}
\def\adjoint{\dagger} 
\def\C {{\bf C}}

\def\L {{\mathcal L}}

\def\cS {{\ensuremath{\cal S}}}

\DeclareMathOperator{\Span}{\operatorname{Span}}
\DeclareMathOperator{\Range}{\operatorname{Range}}

\newcommand{\ADVpm} {\mathrm{Adv}^{\pm}}

\newcommand{\B}{B}	
\def\rtensor {{r\tensor}}

\def\path {\xi}

\DeclareMathOperator{\wsizeop}{{\operatorname{wsize}}}		
\newcommand{\wsize}[1]{{\wsizeop({#1})}}

\newcommand{\wsizeS}[2]{{\wsizeop_{#2}({#1})}}

\newcommand{\wsizexS}[3]{{\wsizeop_{#3}({#1},{#2})}}

\DeclareMathOperator{\wsizefop}{{\operatorname{fwsize}}}	
\newcommand{\wsizef}[1]{{\wsizefop({#1})}}

\newcommand{\wsizefS}[2]{{\wsizefop_{#2}({#1})}}

\newcommand{\wsizefx}[2]{{\wsizefop({#1},{#2})}}
\newcommand{\wsizefxS}[3]{{\wsizefop_{#3}({#1},{#2})}}

\def\Ifree{I_\mathrm{free}}		

\def\jc {\varsigma}	

\DeclareMathOperator{\abst}{\operatorname{abs}}
\def\biadj {B}	

\DeclareMathOperator{\AND}{\ensuremath{\operatorname{AND}}}

\DeclareMathOperator{\OR}{\ensuremath{\operatorname{OR}}}
\DeclareMathOperator{\MAJ}{{\operatorname{MAJ}_3}}

\newcounter{sprows}

\newlength{\spheight}
\newlength{\spraise}

\newlength{\commentslength}

\newcommand{\rem}[1]{}

\newtheorem{theorem}{Theorem}[section]
\newtheorem{lemma}[theorem]{Lemma}
\newtheorem{corollary}[theorem]{Corollary}
\newtheorem{claim}[theorem]{Claim}

\newtheorem{definition}[theorem]{Definition}

\newfont{\subsubsecfnt}{ptmri8t at 10pt}
\renewcommand{\subparagraph}[1]{\smallskip{\subsubsecfnt #1.}}

\numberwithin{equation}{section} 

\newcommand{\eqnref}[1]{\hyperref[#1]{{(\ref*{#1})}}}
\newcommand{\thmref}[1]{\hyperref[#1]{{Theorem~\ref*{#1}}}}
\newcommand{\shortthmref}[1]{\hyperref[#1]{{Thm.~\ref*{#1}}}}
\newcommand{\lemref}[1]{\hyperref[#1]{{Lemma~\ref*{#1}}}}
\newcommand{\corref}[1]{\hyperref[#1]{{Corollary~\ref*{#1}}}}
\newcommand{\defref}[1]{\hyperref[#1]{{Definition~\ref*{#1}}}}
\newcommand{\secref}[1]{\hyperref[#1]{{Section~\ref*{#1}}}}
\newcommand{\figref}[1]{\hyperref[#1]{{Figure~\ref*{#1}}}}
\newcommand{\tabref}[1]{\hyperref[#1]{{Table~\ref*{#1}}}}
\newcommand{\remref}[1]{\hyperref[#1]{{Remark~\ref*{#1}}}}
\newcommand{\appref}[1]{\hyperref[#1]{{Appendix~\ref*{#1}}}}
\newcommand{\claimref}[1]{\hyperref[#1]{{Claim~\ref*{#1}}}}
\newcommand{\propref}[1]{\hyperref[#1]{{Proposition~\ref*{#1}}}}
\newcommand{\exampleref}[1]{\hyperref[#1]{{Example~\ref*{#1}}}}
\newcommand{\conjref}[1]{\hyperref[#1]{{Conjecture~\ref*{#1}}}}

\allowdisplaybreaks[1]

\begin{document}

\title{Faster quantum algorithm for evaluating game trees}
\author{%
Ben W.~Reichardt%
  \thanks{School of Computer Science and Institute for Quantum Computing, University of Waterloo.} 
}
\date{}


\maketitle

\begin{abstract}
We give an $O(\sqrt n \log n)$-query quantum algorithm for evaluating size-$n$ AND-OR formulas.  Its running time is poly-logarithmically greater after efficient preprocessing.  Unlike previous approaches, the algorithm is based on a quantum walk on a graph that is not a tree.  Instead, the algorithm is based on a hybrid of direct-sum span program composition, which generates tree-like graphs, and a novel tensor-product span program composition method, which generates graphs with vertices corresponding to minimal zero-certificates.  

For comparison, by the general adversary bound, the quantum query complexity for evaluating a size-$n$ read-once AND-OR formula is at least $\Omega(\sqrt n)$, and at most $O(\sqrt n \log n / \log \log n)$.  However, this algorithm is not necessarily time efficient; the number of elementary quantum gates applied between input queries could be much larger.  Ambainis et al.\ have given a quantum algorithm that uses $\sqrt n \, 2^{O(\sqrt{\log n})}$ queries, with a poly-logarithmically greater running time.  
\end{abstract}

\section{Introduction}

An AND-OR formula is a rooted tree in which the internal nodes correspond to AND or OR gates.  The size of the formula is the number of leaves.  To a formula $\varphi$ of size $n$ and a numbering of the leaves from $1$ to $n$ corresponds a function $\varphi : \{0,1\}^n \rightarrow \{0,1\}$.  This function is defined on input $x \in \{0,1\}^n$ by placing bit $x_j$ on the $j$th leaf, for $j = 1, 2, \ldots, n$, and evaluating the gates toward the root.   Evaluating an AND-OR formula solves the decision version of a MIN-MAX tree, also known as a two-player game tree.  

Let $Q(\varphi)$ be the quantum query complexity for evaluating the size-$n$ AND-OR formula $\varphi$.  Quantum query complexity is the generalization of classical decision tree complexity to quantum algorithms.  Now the general adversary bound of $\varphi$ is $\ADVpm(\varphi) = \sqrt n$~\cite{BarnumSaks04readonce, HoyerLeeSpalek07negativeadv}, and thus $Q(\varphi) = \Omega(\ADVpm(\varphi)) = \Omega(\sqrt n)$.  Since the general adversary bound is nearly tight for any boolean function, in particular $Q(\varphi) = O(\sqrt n \log n / \log \log n)$~\cite{Reichardt09spanprogram}.  (Interpreted in a different way, this says that the square of the quantum query complexity of evaluating a boolean function is almost a lower bound on the read-once formula size for that function~\cite{LaplanteLeeSzegedy06adversary}.)  However, the algorithm from~\cite{Reichardt09spanprogram} is not necessarily even time efficient.  That is, even though the number of queries to the input is nearly optimal, the number of elementary quantum gates applied between input queries could be much larger.  

Ambainis et al.~\cite{AmbainisChildsReichardtSpalekZhang07andor} have given a quantum algorithm that evaluates $\varphi$ using $\sqrt n \, 2^{O(\sqrt{\log n})}$ queries, with a running time only poly-logarithmically larger after efficient preprocessing.  We reduce the query overhead from $2^{O(\sqrt {\log n})}$ to only $O(\log n)$, with the same preprocessing assumption.   

\begin{theorem} \label{t:andor}
Let $\varphi$ be an AND-OR formula of size $n$.  Then $\varphi$ can be evaluated with error at most $1/3$ by a quantum algorithm that uses $O(\sqrt{n} \log n)$ input queries.  After polynomial-time classical preprocessing independent of the input, and assuming unit-time coherent access to the preprocessed string, the running time of the algorithm is $\sqrt n \, (\log n)^{O(1)}$.  
\end{theorem}

An improvement from $2^{O(\sqrt{\log n})}$ to $O(\log n)$ overhead may not be significant for eventual practical applications.  Additionally, the algorithm does not obviously bring us closer to knowing whether the general adversary bound is tight for quantum query complexity, because its overhead is larger than the $O(\log n / \log \log n)$ overhead of the query algorithm from~\cite{Reichardt09spanprogram}.  (It may be that the analysis used to prove \thmref{t:andor} is somewhat loose.)  

However, the idea behind the algorithm is still of interest, as it provides a new solution to the problem of evaluating AND-OR formulas with large depth.  If $\varphi$ is a formula on $n$ variables, with depth $d$, then the~\cite{AmbainisChildsReichardtSpalekZhang07andor} algorithm, applied directly, evaluates $\varphi$ using $O(\sqrt {n d})$ queries.\footnote{Actually, \cite[Sec.~7]{AmbainisChildsReichardtSpalekZhang07andor} only shows a bound of $O(\sqrt n \, d^{3/2})$ queries, but this can be improved to $O(\sqrt n \, d )$ using the bounds on $\sigma_\pm(\varphi)$ below \cite[Def.~1]{AmbainisChildsReichardtSpalekZhang07andor}.  The improved analysis of the same algorithm in~\cite{Reichardt09unbalancedformula} tightens this to $O(\sqrt {n d})$, which is also the depth-dependence found by~\cite{ambainis07nand}.}  Thus for a highly unbalanced formula, with depth $d = \Omega(n)$, the quantum algorithm performs no better asymptotically than the trivial $n$-query classical algorithm.  Fortunately, Bshouty et al.\ have given a ``rebalancing" procedure that takes AND-OR formula $\varphi$ as input and outputs an equivalent AND-OR formula $\varphi'$ with depth $d' = 2^{O(\sqrt{\log n})}$ and size $n' = n \, 2^{O(\sqrt {\log n})}$~\cite{bce:size-depth, bb:size-depth}.  Appealing to this result, \cite{AmbainisChildsReichardtSpalekZhang07andor} evaluates $\varphi'$ using $O(\sqrt {n' d'} ) = \sqrt n \, 2^{O(\sqrt{\log n})}$ queries.  

The algorithm behind \thmref{t:andor} gets around the large-depth problem without using formula rebalancing.  Instead, the algorithm is based on a novel method for constructing bipartite graphs with certain useful spectral properties.  Ambainis et al.\ run a quantum walk on a graph that matches the formula tree, except with certain edge weights and boundary conditions at the leaves.  This tree comes from glueing together elementary graphs for each gate.  We term this composition method ``direct-sum" composition, because the graph's adjacency matrix acts on a space that is the direct sum of spaces for each individual gate.  Direct-sum composition incurs severe overhead on highly unbalanced formulas, making the query complexity at least proportional to the formula depth.  

The new algorithm begins with the same elementary graphs, with even the same edge weights.  However, it composes them using a kind of ``tensor-product" graph composition method.  Overall, this results in graphs that have much lower depth, although they are not trees.  By carefully combining this method with direct-sum composition, we obtain a graph on which the algorithm runs a quantum walk.  The two approaches are summarized in more detail in \secref{s:andorcomparison} below.  

The general formula-evaluation problem is an ideal example of a recursively defined problem.  The evaluation of a formula is the evaluation of a function, the inputs of which are themselves independent formulas.  As argued in a companion paper~\cite{Reichardt09unbalancedformula}, quantum computers are particularly well suited for evaluating formulas.  Unlike the situation for classical algorithms, for a large class of formulas the optimal quantum algorithm can work following the formula's recursive structure.  Since \thmref{t:andor} does not require AND-OR formula rebalancing, it extends this quantum recursion paradigm.  Besides its conceptual appeal, this may also be important because the effect of rebalancing on the general quantum adversary bound appears to be less tractable for formulas over gate sets beyond AND and OR.  Therefore the rebalancing step that aids~\cite{AmbainisChildsReichardtSpalekZhang07andor} might not be useful to solve the large-depth problem on more general formulas.  The hybrid graph composition method is another tool that might generalize more easily.  

Our algorithm is developed and analyzed using the framework relating span programs and quantum algorithms from~\cite{Reichardt09spanprogram}.  The connection to time-efficient quantum algorithms, especially for evaluating unbalanced formulas over arbitrary fixed, finite gate sets, has been developed further in~\cite{Reichardt09unbalancedformula}.  \tabref{f:formulastable} summarizes some results for the formula-evaluation problem in the classical and quantum models; for a more detailed and inclusive review, see~\cite{Reichardt09unbalancedformula}.  \secref{s:background} below will go over only the history of quantum algorithms for evaluating AND-OR formulas.  

\begin{table}
\begin{center}
\begin{tabular}{|c|c@{~}c|c@{$\!\!\!\!$}c|}
\hline \hline
  &  \multicolumn{2}{c|}{Randomized, zero-error} & \multicolumn{2}{c|}{Quantum bounded-error} \\
Formula $\varphi$ & \multicolumn{2}{c|}{query complexity $R(\varphi)$} & \multicolumn{2}{c|}{query complexity $Q(\varphi)$} \\
\hline
$\OR_n$ & $n$& & $\Theta(\sqrt n)$&\cite{grover:search, BennettBernsteinBrassardVazirani97upper} \\
\hline
Balanced $\AND_2$-$\OR_2$ & $\Theta(n^\alpha)$&\cite{sw:and-or} & $\Theta(\sqrt n)$&\cite{fgg:and-or, AmbainisChildsReichardtSpalekZhang07andor} \\
\hline
Well-balanced AND-OR & tight recursion&\cite{sw:and-or} & & \\
\hline
Approx.-balanced AND-OR & & 
& $\Theta(\sqrt n)$&\cite{AmbainisChildsReichardtSpalekZhang07andor, Reichardt09unbalancedformula} \\
\hline
Arbitrary AND-OR & $\Omega(n^{0.51})$&\cite{HeimanWigderson91generalANDOR} & 
\begin{minipage}[l]{0.7in}\begin{center}$\Omega(\sqrt n)$\\$O(\sqrt n \log n)$\end{center}\end{minipage}&\begin{minipage}[l]{0.7in}\begin{center}\cite{BarnumSaks04readonce}\\ (\shortthmref{t:andor})\end{center}\end{minipage} \\
\hline 
Balanced $\MAJ$ ($n = 3^d$) & $\Omega\big((7/3)^d\big)$, $O(2.654^d)$&\cite{JayramKumarSivakumar03majority} & $\Theta(2^d)$&\cite{ReichardtSpalek08spanprogram} \\
\hline
Balanced over $\cS$
 & 
 & & $\Theta(\ADVpm(\varphi))$&\cite{Reichardt09spanprogram} \\
\hline
Almost-balanced over $\cS$
 & 
 & & $\Theta(\ADVpm(\varphi))$&\cite{Reichardt09unbalancedformula} \\
\hline \hline
\end{tabular}
\end{center}
\caption{Comparison of some classical and quantum query complexity results for formula evaluation.  Here the exponent $\alpha$ is given by $\alpha = \log_2 (\frac{1 + \sqrt{33}}{4}) \approx 0.753$, and $\cS$ is any fixed, finite gate set.  Under certain assumptions, the algorithms' running times are only poly-logarithmically slower.} 
\label{f:formulastable}
\end{table}

\subsection{Idea of the algorithm} \label{s:andorcomparison}

As an example to illustrate the main idea of our algorithm, consider the AND-OR formula $\varphi(x) = \big( [ (x_1 \wedge x_2) \vee x_3 ] \wedge x_4 \big) \vee \big( x_5 \wedge [x_6 \vee x_7] \big)$, where $\wedge$ and $\vee$ denote $\AND$ and $\OR$, respectively.  In \figref{f:formula}, this formula is represented as a tree.  

\begin{figure}
\centering
\begin{tabular}{c@{$\quad$}c}
\subfigure[]{\label{f:formula}\includegraphics[scale=1]{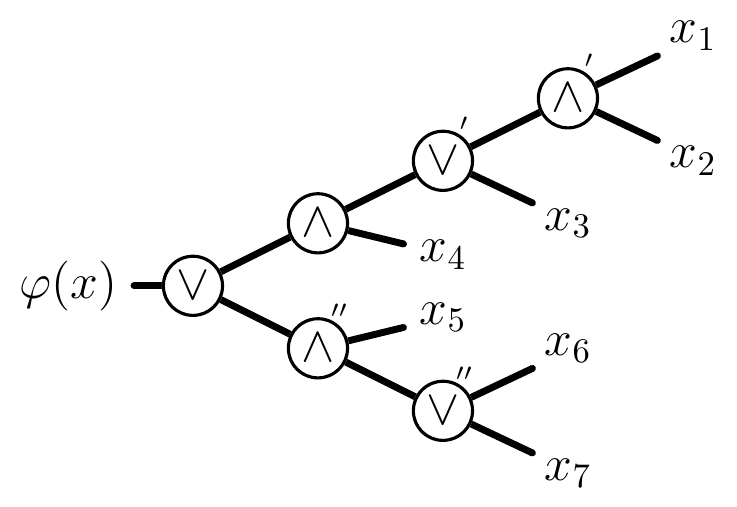}} & 
\subfigure[]{\label{f:acrszgraphextraedges}\includegraphics[scale=1]{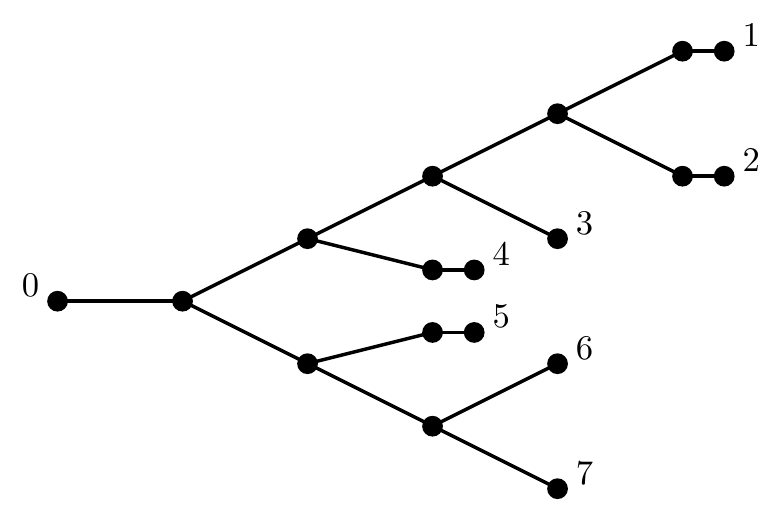}} 
\end{tabular}
\caption{In (a), the AND-OR formula $\varphi(x) = \big( [ (x_1 \wedge x_2) \vee x_3 ] \wedge x_4 \big) \vee \big( x_5 \wedge [x_6 \vee x_7] \big)$ is represented as a tree, with some gates marked for future reference.  Figure~(b) shows the graph used by~\cite{AmbainisChildsReichardtSpalekZhang07andor}.} \label{f:exampledirectsum}
\end{figure}

The \cite{AmbainisChildsReichardtSpalekZhang07andor} algorithm starts with the graph in \figref{f:acrszgraphextraedges}, essentially the same as the formula tree, except with extra edges attached to the root and some leaves.  The edges should be weighted, but for this intuitive discussion take every edge's weight to be one.  

Consider an input $x \in \{0,1\}^7$.  Modify the graph by attached a dangling edge to vertex $j$ if $x_j = 0$, for $j = 1, 2, \ldots, 7$.  Then it is simple to see that the resulting graph has an eigenvalue-zero eigenvector supported on vertex $0$ (or rather its adjacency matrix does) if and only if $\varphi(x) = 1$.  If we added an edge off vertex $0$, then the resulting graph would have an eigenvalue-zero eigenvector supported on the new root if and only if $\varphi(x) = 0$.  

The \cite{AmbainisChildsReichardtSpalekZhang07andor} algorithm takes advantage of this property by running (phase estimation on) a quantum walk that starts at the root vertex.  The algorithm detects the eigenvalue-zero eigenvector in order to evaluate the formula.  

This algorithm does not work well on formulas with large depth.  For example, consider the maximally unbalanced formula on $n$ inputs, a skew tree.  The corresponding graph is nearly the length-$n$ line graph.  It will still have an eigenvalue-zero eigenvector supported on the root if and only if the formula evaluates to zero.  However, the algorithm will require $\Omega(n)$ time to detect this eigenvector, because its squared support on the root is only $O(1/n)$ after normalization, and because the spectral gap around zero will also be $O(1/n)$.  (The spectral gap determines the precision of the phase estimation procedure, and hence its running time.  It corresponds to the squared support of eigenvalue-zero eigenvectors by~\cite[Theorem~8.7]{Reichardt09spanprogram}.)  Alternatively, one can argue that the algorithm requires $\Omega(n)$ time because it takes that long even to \emph{reach} the deepest leaf vertices.  

Now consider the graph in \figref{f:tensorproductgraph}.  Again modify the graph according to an input $x \in \{0,1\}^7$ by attaching dangling edges to those vertices with $x_j = 0$.  Considering a few examples should convince the reader that the resulting graph has an eigenvalue-zero eigenvector supported on vertex $0$ if and only if $\varphi(x) = 0$.  Note, though, that the distance from ``output" vertex $0$ to any of the ``input" vertices $1$ to $7$ is at most two.  The graph is also far from being a tree.  Its main feature is that the ``constraint" vertices---the vertices aside from $0, 1, \ldots, 7$---are in one-to-one correspondence with minimal zero-certificates to $\varphi$.  For example, for $x = 0101011$, $\varphi(x) = 0$, but flipping any bit of $x$ from $0$ to $1$ changes the formula evaluation from $0$ to $1$.  The corresponding constraint vertex is connected to exactly those input vertices $j$ for which $x_j = 0$.  

\begin{figure}
\centering
\begin{tabular}{c@{$\quad$}c}
\subfigure[]{\label{f:tensorproductgraph}\includegraphics[scale=1]{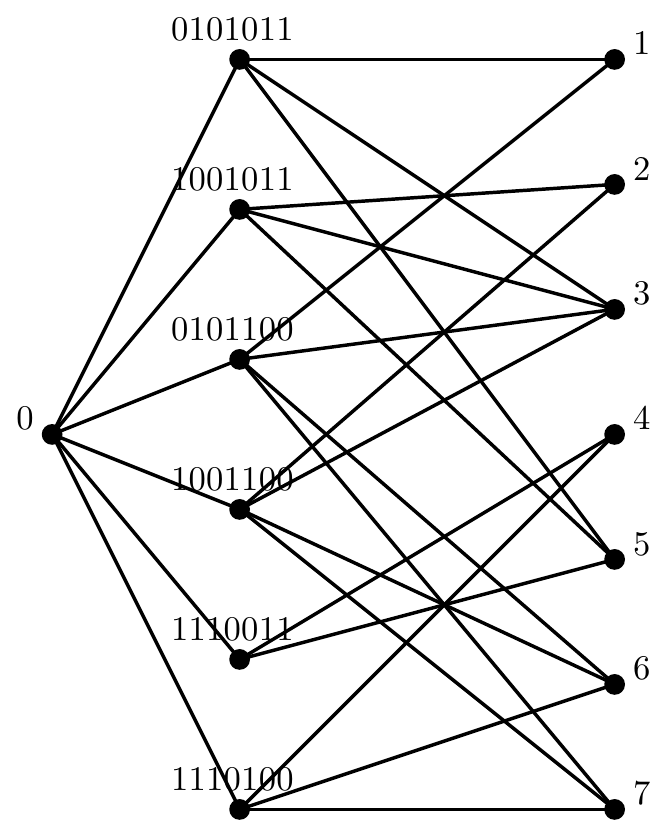}} &
\subfigure[]{\label{f:checkpointedgraph}\includegraphics[scale=1]{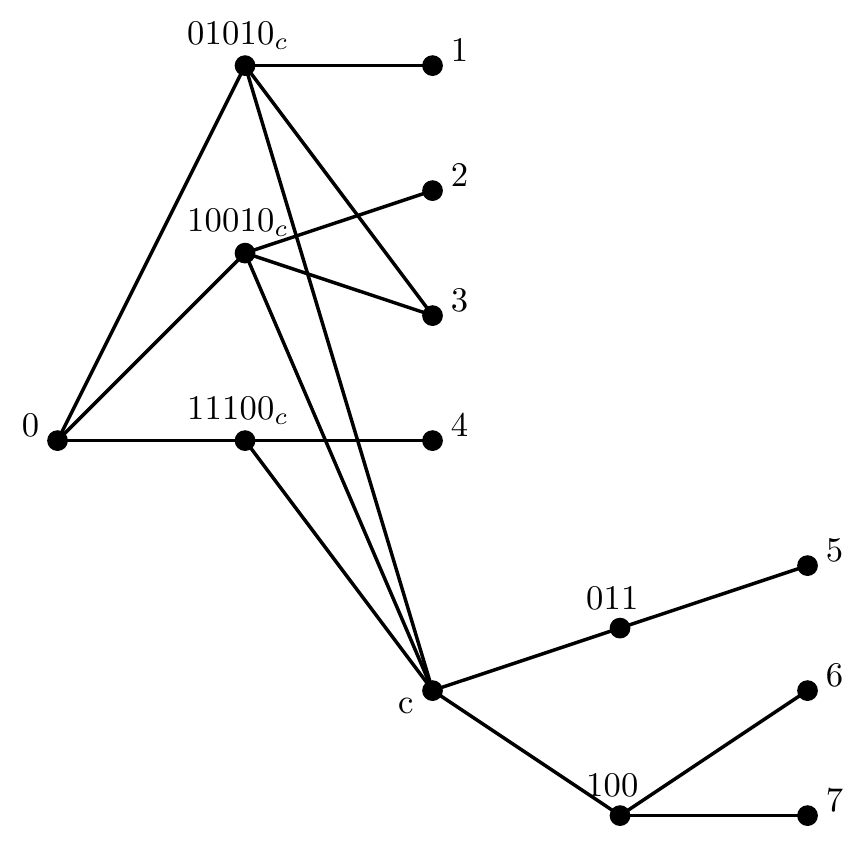}}
\\ \\
\multicolumn{2}{c}{
\subfigure[]{
$\left(\begin{matrix}
\delta \alpha_1 \delta' \alpha_1'  & \epsilon_1 \beta_1 \epsilon_1' \beta_1' & 0 & \epsilon_1 \beta_1 \epsilon_2' \alpha_1' & 0 & \epsilon_2 \alpha_1 \delta' \alpha_1' & 0 & 0 & 0 \\
\delta \alpha_1 \delta' \alpha_2' & 0 & \epsilon_1 \beta_1 \epsilon_1' \beta_2' & \epsilon_1 \beta_1 \epsilon_2' \alpha_2' & 0 & \epsilon_2 \alpha_1 \delta' \alpha_2' & 0 & 0 & 0 \\
\delta \alpha_2 \delta' \alpha' & 0 & 0 & 0 & \epsilon_1 \beta_2 \delta' \alpha' & \epsilon_2 \alpha_2 \delta' \alpha'  & 0 & 0 & 0 \\
0 & 0 & 0 & 0 & 0 & \alpha_1'' \delta'' & \beta_1'' \delta'' & 0 & 0 \\
0 & 0 & 0 & 0 & 0 & \alpha_2'' \delta'' & 0 & \beta_2'' \epsilon_1'' & \beta_2'' \epsilon_2''
\end{matrix}\right)
\place{{\small 0}}{-618mu}{36pt}
\place{{\small 1}}{-537mu}{36pt}
\place{{\small 2}}{-455mu}{36pt}
\place{{\small 3}}{-370mu}{36pt}
\place{{\small 4}}{-290mu}{36pt}
\place{{\small c}}{-210mu}{36pt}
\place{{\small 5}}{-141mu}{36pt}
\place{{\small 6}}{-88mu}{36pt}
\place{{\small 7}}{-37mu}{36pt}
\place{{\small $01010_c$}}{-690mu}{24pt}
\place{{\small $10010_c$}}{-690mu}{12pt}
\place{{\small $11100_c$}}{-690mu}{0pt}
\place{{\small $011$}}{-690mu}{-12pt}
\place{{\small $100$}}{-690mu}{-24pt}
$}
}
\end{tabular}
\caption{
For the formula from \figref{f:formula}, we display the reduced tensor-product composed span program in (a), and the hybrid-composed span program in (b) for the case when the edge to the subformula on inputs $x_5, x_6, x_7$ is checkpointed (see \secref{s:hybridconstruction}).  In each case, the output vertex corresponding to the target vector is labeled $0$ and the input bits are labeled from $1$ to $7$.  The other vertices in (a) are labeled according to the corresponding maximal false inputs (see \defref{t:maximalfalsedef} and \lemref{t:reducedtensorproductANDORcompose}).  In (b), the checkpoint vertex is labeled c---this vertex corresponds to the free input vector added in direct-sum composition (see \defref{t:directsumcomposedef})---and the other vertices have been labeled according to maximal false inputs with the checkpointed subformula grouped together.  In (c), the edge weights for the graph in (b) are specified using the biadjacency matrix.  The parameters $\alpha, \delta, \ldots$ are marked with $'$ or $''$ to indicate to which gate in the formula of \figref{f:formula} they correspond.} \label{f:exampletensorproduct}
\end{figure}

The graphs in Figures~\ref{f:acrszgraphextraedges} and~\ref{f:tensorproductgraph} represent two extremes of a family of graphs that evaluate $\varphi$ in the same manner.  The graph in \figref{f:acrszgraphextraedges} can be seen (essentially) as coming from plugging together along $\varphi$ the individual graphs in Figures~\ref{f:orgate} and~\ref{f:andgate} that evaluate single OR and AND gates in the same manner.  We term this ``direct-sum" composition.  The graph in \figref{f:tensorproductgraph} comes from a certain ``tensor-product" composition of the graphs from Figures~\ref{f:orgate} and~\ref{f:andgate}.  (\figref{f:reducedtensorproductcompositionexamples} shows several other examples of tensor-product graph composition.)  

With the correct choice of edge weights, tensor-product composition leads to graphs for which the squared support of the unit-normalized eigenvalue-zero eigenvector on the root is $\Omega(1/\sqrt n)$ for a formula of size $n$.  This implies by~\cite[Theorem~9.1]{Reichardt09spanprogram} a quantum algorithm that uses $O(\sqrt n \log n / \log \log n)$ queries to evaluate the formula.  There is no issue with deep formulas.  However, the algorithm will not be time efficient.  Essentially, the problem is that the number of vertices can be exponentially large in $n$, as can be their degrees, which makes it difficult to implement a quantum walk efficiently.  

\thmref{t:andor} gets around this problem by using a combination of the two methods of composing graphs.  For example, the graph in \figref{f:checkpointedgraph} also evaluates the formula $\varphi$ in the same manner.  One can think of the vertex $c$ as evaluating the subformula $x_5 \wedge (x_6 \vee x_7)$.  The subgraph it cuts off has been composed in a direct-sum manner with a tensor-product-composed graph for the formula $([(x_1 \wedge x_2) \vee x_3] \wedge x_4) \vee x_c$.  By combining the two composition methods, we can manage the tradeoffs, controlling the maximum degree and norm of the graph, while also avoiding the formula depth problem.  

Although our algorithm can be presented and analyzed entirely in terms of graphs, we will present it in terms of span programs.  Span programs are part of a framework for designing and analyzing quantum algorithms~\cite{Reichardt09spanprogram}, for which \secref{s:spanprograms} gives some necessary background.  Eigenvalue-zero eigenvectors correspond to span program witnesses and the squared support on the root corresponds to a complexity measure known as the full witness size~\cite{Reichardt09unbalancedformula}.  The two graph composition techniques described above correspond to different ways of composing general span programs.

\subsection{Review of quantum algorithms for evaluating AND-OR formulas} \label{s:background}

Research on the formula-evaluation problem in the quantum model began with the simple $n$-bit OR function, $\OR_n$.  Grover gave a quantum algorithm for evaluating $\OR_n$ with bounded one-sided error using $O(\sqrt n)$ oracle queries and $O(\sqrt n \log \log n)$ time~\cite{grover:search, grover:search-time}.  

Grover's algorithm, together with repetition for reducing errors, can be applied recursively to speed up the evaluation of more general AND-OR formulas.  For example, the size-$n$ AND-OR formula $\AND_{\sqrt n} \circ ( \OR_{\sqrt n}, \ldots, \OR_{\sqrt n} )$ can be evaluated in $O(\sqrt n \log n)$ queries.  Here the extra logarithmic factor comes from using repetition to reduce the error probability of the inner $\OR_{\sqrt n}$ evaluation procedure from a constant to be polynomially small.  Call a formula \emph{layered} if the gates at the same depth are the same.  Buhrman, Cleve and Wigderson show that the above argument can be applied to evaluate a layered, depth-$d$ AND-OR formula on $n$ inputs using $O(\sqrt n \log^{d-1} n)$ queries~\cite[Theorem~1.15]{BuhrmanCleveWigderson98}.  

H\o yer, Mosca and de Wolf \cite{HoyerMoscaDeWolf03berror-search} consider the case of a unitary input oracle $\tilde O_x$ that maps 
\beq
\tilde O_x: \, \ket \varphi \otimes \ket j \otimes \ket b \otimes \ket 0
\mapsto 
\ket \varphi \otimes \ket j \otimes \big( \ket{b \oplus x_j} \otimes \ket{\psi_{x,j,x_j}} + \ket{b \oplus \overline x_j} \otimes \ket{\psi_{x,j,\overline x_j}} \big)
 \enspace ,
\eeq
where $\ket{\psi_{x,j,x_j}}$, $\ket{\psi_{x,j,\overline x_j}}$ are pure states with $\norm{\ket{\psi_{x,j,x_j}}}^2 \geq 2/3$.  Such an oracle can be implemented when the function $j \mapsto x_j$ is computed by a bounded-error, randomized subroutine~\cite{NielsenChuang00}.  H\o yer et al.\ allow access to $\tilde O_x$ and $\tilde O_x^{-1}$, both at unit cost, and show that $\OR_n$ can still be evaluated using $O(\sqrt n)$ queries.  This \emph{robustness} result implies that the $\log n$ steps of repetition used by~\cite{BuhrmanCleveWigderson98} are not necessary, and a depth-$d$ layered AND-OR formula can be computed in $O(\sqrt n \, c^{d-1})$ queries, for some constant $c > 1000$.  This gives an $O(\sqrt n)$-query quantum algorithm for the case that the depth $d$ is constant, but is not sufficient to cover, e.g., the complete, binary AND-OR formula, for which $d = \log_2 n$.  

A breakthrough for the formula-evaluation problem came in 2007, when Farhi, Goldstone and Gutmann presented a quantum algorithm for evaluating complete, binary AND-OR formulas~\cite{fgg:and-or}.  Their algorithm is not based on iterating Grover's algorithm in any way, but instead runs a quantum walk---analogous to a classical random walk---on a graph derived from the AND-OR formula graph as in \figref{f:exampledirectsum}.  The algorithm runs in time $O(\sqrt n)$ in a certain continuous-time query model.  

Ambainis et al.\ discretized the~\cite{fgg:and-or} algorithm by reinterpreting a correspondence between discrete-time random and quantum walks due to Szegedy~\cite{Szegedy04walkfocs} as a correspondence between \emph{continuous-time} and discrete-time quantum walks~\cite{AmbainisChildsReichardtSpalekZhang07andor}.  Applying this correspondence to quantum walks on certain \emph{weighted} graphs, they gave an $O(\sqrt n)$-query quantum algorithm for evaluating ``approximately balanced" formulas, extended in~\cite{Reichardt09unbalancedformula}.  Using the formula rebalancing procedure of~\cite{bce:size-depth, bb:size-depth}, the~\cite{AmbainisChildsReichardtSpalekZhang07andor} algorithm uses $\sqrt n 2^{O(\sqrt{\log n})}$ queries in general.  This is nearly optimal, since the adversary bound gives an $\Omega(\sqrt n)$ lower bound~\cite{BarnumSaks04readonce}.  

This author has given an $O(\sqrt n \log n / \log \log n)$-query quantum algorithm for evaluating arbitrary size-$n$ AND-OR formulas~\cite{Reichardt09spanprogram}.  In fact, the result is more general, stating that the general adversary bound is nearly tight for \emph{every} boolean function.  However, unlike the earlier AND-OR formula-evaluation algorithms, the algorithm is not necessarily time efficient.

\section{Span programs} \label{s:spanprograms}

In this section, we will briefly recall some of the definitions and results on span programs from~\cite{Reichardt09spanprogram, Reichardt09unbalancedformula}.  This section is essentially an abbreviated version of~\cite[Sec.~2]{Reichardt09unbalancedformula}.  

For a natural number $n$, let $[n] = \{1, 2, \ldots, n\}$.  For a finite set $X$, let $\C^X$ be the inner product space $\C^{\abs X}$ with orthonormal basis $\{ \ket x : x \in X \}$.  For vector spaces $V$ and $W$ over $\C$, let $\L(V, W)$ be the set of linear transformations from $V$ into $W$, and let $\L(V) = \L(V, V)$.  For $A \in \L(V, W)$, $\norm{A}$ is its operator norm.  Let $\B = \{0,1\}$.  For a string $x \in \B^n$, let $\bar x$ denote its bitwise complement.

\subsection{Span program full witness size} \label{s:spanprogramdef}

The full witness size is a span program complexity measure that is important for developing quantum algorithms that are time efficient as well as query efficient.  

\begin{definition}[Span program~\cite{KarchmerWigderson93span}] \label{t:spanprogramdef}
A span program $P$ consists of a natural number $n$, a finite-dimensional inner product space $V$ over $\C$, a ``target" vector $\ket t \in V$, disjoint sets $\Ifree$ and $I_{j,b}$ for $j \in [n]$, $b \in \B$, and ``input vectors" $\ket{v_i} \in V$ for $i \in \Ifree \cup \bigcup_{j \in [n], b \in \B} I_{j,b}$.  

To $P$ corresponds a function $f_P : \B^n \rightarrow \B$, defined on $x \in \B^n$ by 
\beq \label{e:spanprogramdef}
f_P(x) = \begin{cases}
1 & \text{if $\ket t \in \Span(\{ \ket{ v_i } : i \in \Ifree \cup \bigcup_{j \in [n]} I_{j, x_j} \})$} \\ 
0 & \text{otherwise}
\end{cases}
\eeq
\end{definition}

Some additional notation is convenient.  Fix a span program $P$.  Let $I = \Ifree \cup \bigcup_{j \in [n], b \in \B} I_{j,b}$.  Let $A \in \L(\C^I, V)$ be given by $A = \sum_{i \in I} \ketbra{v_i}{i}$.  For $x \in \B^n$, let $I(x) = \Ifree \cup \bigcup_{j \in [n]} I_{j, x_j}$ and $\Pi(x) = \sum_{i \in I(x)} \ketbra i i \in \L(\C^I)$.  Then $f_P(x) = 1$ if $\ket t \in \Range(A \Pi(x))$.  A vector $\ket w \in \C^I$ is said to be a witness for $f_P(x) = 1$ if $\Pi(x) \ket w = \ket w$ and $A \ket w = \ket t$.  A vector $\ket{w'} \in V$ is said to be a witness for $f_P(x) = 0$ if $\braket t {w'} = 1$ and $\Pi(x) A^\dagger \ket{w'} = 0$.  

\begin{definition}[Witness size] \label{t:wsizedef}
Consider a span program $P$, and a vector $s \in [0, \infty)^n$ of nonnegative ``costs."  Let $S = \sum_{j \in [n], b \in \B, i \in I_{j,b}} \sqrt{s_j} \ketbra i i \in \L(\C^I)$.  For each input $x \in \B^n$, define the witness size of $P$ on $x$ with costs $s$, $\wsizexS P x s$, as follows: 
\beq \label{e:wsizedef}
\wsizexS P x s = \begin{cases}
\min_{\ket w : \, A \Pi(x) \ket w = \ket t} \norm{S \ket w}^2 & \text{if $f_P(x) = 1$} \\
\min_{\substack{\ket{w'} : \, \braket{t}{w'} = 1 \\ \Pi(x) A^\adjoint \ket{w'} = 0}} \norm{S A^\adjoint \ket{w'}}{}^2 & \text{if $f_P(x) = 0$}
\end{cases}
\eeq

The witness size of $P$ with costs $s$ is 
\beq
\wsizeS P s =  \max_{x \in \B^n} \wsizexS P x s
 \enspace .
\eeq

Define the full witness size $\wsizefS P s$ by letting $S^f = S + \sum_{i \in \Ifree} \ketbra i i$ and 
\begin{align}
\wsizefxS P x s &= \begin{cases}
\min_{\ket w : \, A \Pi(x) \ket w = \ket t} (1 + \norm{S^f \ket w}{}^2) & \text{if $f_P(x) = 1$} \\
\min_{\substack{\ket{w'} : \, \braket{t}{w'} = 1 \\ \Pi(x) A^\adjoint \ket{w'} = 0}} (\norm{\ket{w'}}{}^2 + \norm{S A^\adjoint \ket{w'}}{}^2) & \text{if $f_P(x) = 0$}
\end{cases} \\
\wsizefS P s &=  \max_{x \in \B^n} \wsizefxS P x s
 \enspace .
\end{align}
\end{definition}

When the subscript $s$ is omitted, the costs are taken to be uniform, $s = \vec 1 = (1, 1, \ldots, 1)$, e.g., $\wsizef P = \wsizefS P {\vec 1}$.  The witness size is defined in~\cite{ReichardtSpalek08spanprogram}.  The full witness size is defined in~\cite[Sec.~8]{Reichardt09spanprogram}, but is not named there.  A \emph{strict} span program has $\Ifree = \emptyset$, so $S^f = S$, and a \emph{monotone} span program has $I_{j,0} = \emptyset$ for all $j$~\cite[Def.~4.9]{Reichardt09spanprogram}.  

\begin{theorem}[{\cite[Theorem~9.3]{Reichardt09spanprogram}, \cite[Theorem~2.3]{Reichardt09unbalancedformula}}] \label{t:generalspanprogramalgorithmnonblackbox}
Let $P$ be a span program.  Then $f_P$ can be evaluated using 
\beq \label{e:generalspanprogramalgorithmnonblackbox}
T = O\big( \wsizef P \, \norm{\abst(A_{G_P})} \big)
\eeq
quantum queries, with error probability at most $1/3$.  Moreover, if the maximum degree of a vertex in $G_P$ is $d$, then the time complexity of the algorithm for evaluating $f_P$ is at most a factor of $(\log d) \big(\log (T \log d)\big)^{O(1)}$ worse, after classical preprocessing and assuming constant-time coherent access to the preprocessed string.  
\end{theorem}

\subsection{Direct-sum span program composition}

Let $f : \B^n \rightarrow \B$ be a boolean function.  Let $S \subseteq [n]$.  For $j \in [n]$, let $m_j$ be a natural number, with $m_j = 1$ for $j \notin S$.  For $j \in S$, let $f_j : \B^{m_j} \rightarrow \B$.  Define $y : \B^{m_1} \times \cdots \times \B^{m_n} \rightarrow \B^n$ by 
\beq
y(x)_j = \begin{cases} f_j(x_j) & \text{if $j \in S$} \\ x_j & \text{if $j \notin S$} \end{cases}
\eeq  
Define $g : \B^{m_1} \times \cdots \times \B^{m_n} \rightarrow \B$ by $g(x) = f(y(x))$.  Given span programs for the individual functions $f$ and $f_j$ for $j \in S$, we will construct a span program for $g$.  

Let $P$ be a span program computing $f_P = f$.  Let $P$ have inner product space $V$, target vector $\ket t$ and input vectors $\ket{v_i}$ indexed by $\Ifree$ and $I_{jc}$ for $j \in [n]$ and $c \in \B$.  

For $j \in [n]$, let $s_j \in [0, \infty)^{m_j}$ be a vector of costs, and let $s \in [0, \infty)^{\sum m_j}$ be the concatenation of the vectors $s_j$.  For $j \in S$, let $P^{j0}$ and $P^{j1}$ be span programs computing $f_{P^{j1}} = f_j : \B^{m_j} \rightarrow \B$ and $f_{P^{j0}} = \neg f_j$, with $r_j = \wsizeS{P^{j0}}{s_j} = \wsizeS{P^{j1}}{s_j}$.  For $c \in \B$, let $P^{jc}$ have inner product space $V^{jc}$ with target vector $\ket{t^{jc}}$ and input vectors indexed by $\Ifree^{jc}$ and $I^{jc}_{kb}$ for $k \in [m_j]$, $b \in \B$.  For $j \notin S$, let $r_j = s_j$.  

Let $I_S = \bigcup_{j \in S, c \in \B} I_{jc}$.  Define $\jc : I_S \rightarrow [n] \times \B$ by $\jc(i) = (j,c)$ if $i \in I_{jc}$.  The idea is that $\jc$ maps $i$ to the input span program that must evaluate to true in order for $\ket{v_i}$ to be available for $P$.  

\begin{definition}[{\cite[Def.~4.5]{Reichardt09spanprogram}}] \label{t:directsumcomposedef}
The direct-sum-composed span program $Q^\oplus$ is defined by: 
\begin{itemize}
\item
The inner product space is $V^\oplus = V \oplus \bigoplus_{j \in S, c \in \B} (\C^{I_{jc}} \otimes V^{jc})$.  Any vector in $V^\oplus$ can be uniquely expressed as $\ket{u}_V + \sum_{i \in I_S} \ket{i} \otimes \ket{u_i}$, where $\ket u \in V$ and $\ket{u_i} \in V^{\jc(i)}$.  
\item
The target vector is $\ket{t^\oplus} = \ket{t}_V$.  
\item
The free input vectors are indexed by $\Ifree^\oplus = \Ifree \cup I_S \cup \bigcup_{j \in S, c \in \B} (I_{jc} \times \Ifree^{jc})$ with, for $i \in \Ifree^\oplus$, 
\beq
\ket{v^\oplus_i} = \begin{cases}
\ket{v_i}_V & \text{if $i \in \Ifree$} \\
\ket{v_i}_V + \ket i \otimes \ket{t^{jc}} & \text{if $i \in I_{jc}$ and $j \in S$} \\		
\ket{i'} \otimes \ket{v_{i''}} & \text{if $i = (i', i'') \in I_{jc} \times \Ifree^{jc}$}
\end{cases}
\eeq
\item
The other input vectors are indexed by $I^\oplus_{(jk)b}$ for $j \in [n]$, $k \in [m_j]$, $b \in \B$.  For $j \notin S$, $I^\oplus_{(j1)b} = I_{jb}$, with $\ket{v^\oplus_i} = \ket{v_i}_V$ for $i \in I^\oplus_{(j1)b}$.  For $j \in S$, let $I^\oplus_{(jk)b} = \bigcup_{c \in \B} (I_{jc} \times I^{jc}_{kb})$.  For $i \in I_{jc}$ and $i' \in I^{jc}_{kb}$, let 
\beq
\ket{v^\oplus_{ii'}} = \ket{i} \otimes \ket{v_{i'}}
 \enspace .
\eeq
\end{itemize}
\end{definition}

By~\cite[Theorem~4.3]{Reichardt09spanprogram}, $f_{Q^\oplus} = g$ and $\wsizeS {Q^\oplus} s \leq \wsizeS P r$.  Moreover, we can bound how quickly the full witness size can grow relative to the witness size: 

\begin{lemma}[{\cite[Lemma~2.5]{Reichardt09unbalancedformula}}] \label{t:wsizefcompose}
Under the above conditions, for each input $x \in \B^{m_1} \times \cdots \times \B^{m_n}$, with $y = y(x)$, 
\begin{itemize}
\item 
If $g(x) = 1$, let $\ket w$ be a witness to $f_P(y) = 1$ with $\sum_{j \in [n], i \in I_{j y_j}} r_j \abs{w_i}^2 = \wsizexS P y r$.  Then 
\beq\begin{split} \label{e:wsizefcomposetrue}
\frac{ \wsizefxS {Q^\oplus} x s }{ \wsizexS P y r }
&\leq 
\sigma\big(y, \ket w\big) + \frac{1 + \sum_{i \in \Ifree} \abs{w_i}^2}{\wsizexS P y r} \\
&\text{where $\sigma(y, \ket w) = \max_{\substack{j \in S : \\ \text{$\exists i \in I_{j y_j}$ with $\braket{i}{w} \neq 0$}}} \frac{\wsizefS {P^{j y_j}} {s_j}}{\wsizeS {P^{j y_j}} {s_j}}$}
 \enspace .
\end{split}\eeq
\item 
If $g(x) = 0$, let $\ket{w'}$ be a witness to $f_P(y) = 0$ such that $\sum_{j \in [n], i \in I_{j \bar y_j}} r_j \abs{\braket{w'}{v_i}}^2 = \wsizexS P y r$.  Then 
\beq\begin{split} \label{e:wsizefcomposefalse}
\frac{ \wsizefxS {Q^\oplus} x s }{ \wsizexS P y r }
&\leq 
\sigma(\bar y, \ket{w'}) + \frac{\norm{\ket{w'}}^2}{\wsizexS P y r} \\
&\text{where $\sigma(\bar y, \ket{w'}) = \max_{\substack{j \in S : \\ \text{$\exists i \in I_{j \bar y_j}$ with $\braket{v_i}{w'} \neq 0$}}} \frac{\wsizefS {P^{j \bar y_j}} {s_j}}{\wsizeS {P^{j \bar y_j}} {s_j}}$}
 \enspace .
\end{split}\eeq
\end{itemize}
If $S = \emptyset$, then $\sigma(y, \ket w)$ and $\sigma(\bar y, \ket{w'})$ should each be taken to be $1$ in the above equations.  
\end{lemma}

\begin{lemma}[{\cite[Lemma~3.4]{Reichardt09unbalancedformula}}] \label{t:directsumnorm}
If $P_\varphi$ is the direct-sum composition along a formula $\varphi$ of span programs $P_v$ and $P_v^\dagger$, then 
\beq
\norm{\abst(A_{G_P})} \leq 2 \max_{v \in \varphi} \max \{ \norm{\abst(A_{G_{P_v}})}, \norm{\abst(A_{G_{P_v^\dagger}})} \}
 \enspace .
\eeq
If the span programs $P_v$ are monotone, then $\norm{\abst(A_{G_P})} \leq 2 \max_v \norm{\abst(A_{G_{P_v}})}$.  
\end{lemma}

\section{Evaluation of arbitrary AND-OR formulas}

Let $\varphi$ be an AND-OR formula.  In this section, we will prove \thmref{t:andor} by applying \thmref{t:generalspanprogramalgorithmnonblackbox} to a certain composed span program $P_\varphi$.  The construction of $P_\varphi$ has three steps that we will explain in detail below.  
\begin{enumerate}
\item
Eliminate any AND or OR gates with fan-in one, and expand out AND and OR gates with higher fan-ins into gates of fan-in exactly two.   
\item
Mark a certain subset of the edges of the formula.  We call marked edges ``checkpoints."  
\item
Starting with the span programs for $\AND_2$ and $\OR_2$, compose them for the subformulas cut off by checkpointed edges using a version of tensor-product composition.  Compose the resulting span programs across checkpoints using direct-sum composition to yield $P_\varphi$.  
\end{enumerate}

Direct-sum span program composition keeps the norm of the corresponding graph's adjacency matrix under control, as well as the maximum degree of a vertex in the graph.  However, it makes the full witness size grow much larger than the witness size, especially for highly unbalanced formulas.  Reduced tensor-product composition generates strict span programs, for which the full witness size stays close to the witness size.  However, it allows the norm and the maximum degree of the corresponding graph to grow exponentially quickly.  By using both techniques in careful combination, we are able to manage this tradeoff so that \thmref{t:generalspanprogramalgorithmnonblackbox} can be applied.

\secref{s:andorspanprogramdef} presents the span programs we use for fan-in-two AND and OR gates.  

In \secref{s:reducedtensorproductANDORcompose}, we study reduced tensor-product composition for the span programs for AND and OR gates.  Reduced tensor-product composition is a version of tensor-product composition that parsimoniously uses fewer dimensions when possible.  For AND-OR formulas, it has the advantage that the output span program's inner product space bears a close connection to false inputs of the formula $\varphi$, similarly to canonical span programs.  In order to motivate the checkpointing idea, we explain the problems of a scheme based only on reduced tensor-product composition.  

\secref{s:hybridconstruction} then presents our full construction of $P_\varphi$, based on a combination of direct-sum and reduced tensor-product composition.   

\secref{s:andoranalysis} contains the analysis of the graphs $G_{P_\varphi}(x)$ required to apply \thmref{t:generalspanprogramalgorithmnonblackbox}.

\subsection{Span programs for AND$_2$ and OR$_2$} \label{s:andorspanprogramdef}

We will use the following strict, monotone span programs for fan-in-two AND and OR gates: 

\begin{definition}[{\cite[Def.~4.1]{Reichardt09unbalancedformula}}] \label{t:andorspanprogramdef}
For $s_1, s_2 > 0$, define span programs $P_{\AND}(s_1, s_2)$ and $P_{\OR}(s_1, s_2)$ computing $\AND_2$ and $\OR_2$, $\B^2 \rightarrow \B$, respectively, by 
\begin{align}
P_{\AND}(s_1, s_2): &&
\ket t &= \left( \begin{matrix} \alpha_1 \\ \alpha_2 \end{matrix} \right) ,\; 
&\ket{v_1} &= \left( \begin{matrix} \beta_1 \\ 0 \end{matrix} \right) ,\; &\ket{v_2} &= \left( \begin{matrix} 0 \\ \beta_2 \end{matrix} \right) \\
P_{\OR}(s_1, s_2): &&
\ket t &= \delta ,\; &\ket{v_1} &= \epsilon_1 ,\;& \ket{v_2} &= \epsilon_2 
\end{align}
Both span programs have $I_{1,1} = \{1\}$, $I_{2,1} = \{2\}$ and $\Ifree = I_{1,0} = I_{2,0} = \emptyset$.  Here the parameters $\alpha_j, \beta_j, \delta, \epsilon_j$, for $j \in [2]$, are given by 
\begin{align}
\alpha_j &= (s_j / s_p)^{1/4} & \beta_j &= 1 \\
\delta &= 1 & \epsilon_j &= (s_j / s_p)^{1/4}
 \enspace ,
\end{align}
where $s_p = s_1 + s_2$.  
Let $\alpha = \sqrt{\alpha_1^2 + \alpha_2^2}$ and $\epsilon = \sqrt{\epsilon_1^2 + \epsilon_2^2}$.  
\end{definition}

Note that $\alpha, \epsilon \in (1, 2^{1/4}]$.  They are largest when $s_1 = s_2$.  

\begin{claim}[{\cite[Claim~4.2]{Reichardt09unbalancedformula}}] \label{t:andorspanprogramwsize}
The span programs $P_{\AND}(s_1, s_2)$ and $P_{\OR}(s_1, s_2)$ satisfy: 
\begin{align}\begin{split}
\wsizexS{P_{\AND}}{x}{(\sqrt{s_1}, \sqrt{s_2})}
&=
\begin{cases}
\sqrt{s_p} & \text{if $x \in \{11, 10, 01\}$} \\
\frac{\sqrt{s_p}}{2} & \text{if $x = 00$}
\end{cases} \\
\wsizexS{P_{\OR}}{x}{(\sqrt{s_1}, \sqrt{s_2})}
&=
\begin{cases}
\sqrt{s_p} & \text{if $x \in \{00, 10, 01\}$} \\
\frac{\sqrt{s_p}}{2} & \text{if $x = 11$}
\end{cases}
\end{split}\end{align}
\end{claim}

\subsection{Reduced tensor-product span program composition for AND-OR formulas} \label{s:reducedtensorproductANDORcompose}

Reduced tensor-product composition of span programs is introduced in~\cite[Def.~4.6]{Reichardt09spanprogram}.  We repeat the definition here, specializing to the case of monotone, strict span programs acting on disjoint inputs.  Also, for simplicity we consider the case of composing on one span program at a time.  After stating the definition, we specialize further to AND-OR formulas, and characterize the reduced tensor-product span program composition of the AND and OR span programs from \defref{t:andorspanprogramdef}.  

Consider monotone functions $f : \B \times \B^n \rightarrow \B$ and $f' : \B^m \rightarrow \B$.  Let $g : \B^m \times \B^n \rightarrow \B$ be defined by 
\beq
g(x, y) = f\big( f'(x), y \big)
\eeq
for $x \in \B^m$, $y \in \B^n$.  
Let $P$ be a span program computing $f_P = f$ and let $P'$ be a span program computing $f_{P'} = f'$.  Assume that $P$ and $P'$ are both monotone, strict span programs.  

Let span program $P$ be in inner product space $V = \C^{[d]}$, with target vector $\ket t$ and input vectors $\{\ket{v_i}\}$ indexed by $I_{j1}$ for $j \in [n]$.  Let $P'$ be in the inner product space $V'$ with target vector $\ket{t'}$ and input vectors $\{\ket{v'_{i'}}\}$ indexed by $I'_{k1}$ for $k \in [m]$.  Since $P$ and $P'$ are both monotone, $I_{j0} = I'_{k0} = \emptyset$.  

\begin{definition}[\cite{Reichardt09spanprogram}] \label{t:reducedtensorproductcomposedef}
The tensor-product-composed span program, reduced with respect to the basis $\{ \ket l : l \in [d] \}$ for $V$, is $Q^\rtensor$, defined by: 
\begin{itemize}
\item
Let $Z = \{ l \in [d] : \forall \, i \in I_{11}, \braket l {v_i} = 0 \}$.  For $l \in [d]$, define $V_l$ and $\ket{\pi_l} \in V_l$ by  
\beq\begin{split}
V_l &= \begin{cases}
V' & \text{if $l \notin Z$, i.e., $\braket l {v_i} \neq 0$ for some $i \in I_{11}$} \\
\C & \text{if $l \in Z$}
\end{cases} \\
\ket{\pi_l} &= \begin{cases}
\ket{t'} & \text{if $l \notin Z$} \\
\norm{\ket{t'}} & \text{if $l \in Z$} \\
\end{cases}
\end{split}\eeq
\item
The inner product space of $Q^\rtensor$ is $V^\rtensor = \bigoplus_{l \in [d]} V_l$.  Any vector $\ket v \in V^\rtensor$ can be uniquely expressed as $\sum_{l \in [d]} \ket{v_l}_{V_l}$, where $\ket{v_l} \in V_l$.  
\item
The target vector is 
\beq \label{e:reducedtensorproductcomposedeftarget}
\ket{t^\rtensor} = \sum_{l \in [d]} \braket l t \ket{\pi_l}_{V_l}
 \enspace .
\eeq
\item
$Q^\rtensor$ is strict and monotone, thus $\Ifree^\rtensor = I^\rtensor_{k0} = \emptyset$ for all $k \in [m+n]$.  
\item
The input vectors $\{\ket{v^\rtensor_\iota}\}$ are indexed by 
\beq
I^\rtensor_{k1} = \begin{cases}
I_{11} \times I'_{k1} & \text{if $k \leq m$} \\
I_{j1} & \text{if $k > m$, where $j = k-m+1$}
\end{cases}
\eeq
and given by  
\beq \label{e:reducedtensorproductcomposedefinput}
\ket{v^\rtensor_\iota} = \begin{cases}
\sum_{l \in [d]} \braket{l}{v_i} \ket{v'_{i'}}_{V_l} & \text{if $\iota = (i, i') \in \bigcup_{k \leq m} I^\rtensor_{k1}$} \\
\sum_{l \in [d]} \braket{l}{v_\iota} \ket{\pi_l}_{V_l} & \text{if $\iota \in \bigcup_{k > m} I^\rtensor_{k1}$}
\end{cases}
\eeq
\end{itemize}
\end{definition}

The intuition behind this construction is to compose the span programs in a tensor-product manner somewhat similar to \defref{t:directsumcomposedef}.  From~\cite[Def.~4.4]{Reichardt09spanprogram}, this would give a span program with target vector $\ket t \otimes \ket{t'} \in V \otimes V'$ and input vectors either $\ket{v_i} \otimes \ket{v'_{i'}}$ for $i \in I_{11}$ or $\ket{v_i} \otimes \ket{t'}$ otherwise.  However, if all the $I_{11}$ input vectors are zero in a coordinate $l \in [d]$, then the first type of input vectors are all zero on $\ket l \otimes V'$.  The second type of input vectors are all proportional to the same state $\ket l \otimes \ket{t'}$ on that coordinate, so we might as well just keep $\ket l \norm{\ket{t'}}$ as in the above definition.  The advantage over tensor-product composition is that the graph $G_{Q^\rtensor}$ associated to the span program $Q^\rtensor$ potentially has fewer vertices, with lower degrees.  

Here we have composed the span program $P_1$ into the first input of $P$.  Reduced tensor-product composition into the other inputs is defined symmetrically.  

By~\cite[Prop.~4.7]{Reichardt09spanprogram}, for arbitrary costs $r \in [0, \infty)^m$ and $s \in [0, \infty)^n$, 
\beq
\wsizeS {Q^\rtensor} {(r, s)} \leq \wsizeS P {(\wsizeS {P'} r, s)}
 \enspace .
\eeq

\bigskip

Now let us study reduced tensor-product composition for AND-OR formulas.  To start with, it will be helpful to give two examples of \defref{t:reducedtensorproductcomposedef}, for the cases $P' = P_{\OR}$ and $P' = P_{\AND}$.  Recall from~\cite[Def.~8.1]{Reichardt09spanprogram} that the biadjacency matrix for the bipartite graph $G_P(1^{1+n})$ is given by $\biadj_{G_P(1^{1+n})} = ( \begin{matrix} \ket t & A \end{matrix})$, where $A$ is the matrix whose columns are the input vectors of $P$, as defined in \secref{s:spanprogramdef}.  Assume that $I_{11} = \{1\}$ is a singleton set, with $\ket{v_1}$ the first column of $A$.  Rearrange the rows so that the nonzero entries of $\ket{v_1}$ come first (the set $Z$ from \defref{t:reducedtensorproductcomposedef} comes last), writing $\ket{v_1} = (\ket{\gamma}, 0)$, where $\ket{\gamma}$ is nonzero in every entry.  Writing $\ket t = (\ket{t_1}, \ket{t_2})$, $B_{G_P(1^{1+n})}$ decomposes as 
\beq \label{e:reducedtensorproductcomposestart}
B_{G_P(1^{1+n})} = \left(\begin{matrix}
\ket{t_1} & \ket{\gamma} & C_1 \\
\ket{t_2} & 0 & C_2
\end{matrix}\right)
 \enspace .
\eeq
\begin{itemize}
\item
First consider the case that $P' = P_{\AND}(s_1, s_2)$ from \defref{t:andorspanprogramdef}.  Let $Q^\rtensor$ be the composed span program from \defref{t:reducedtensorproductcomposedef}.  Then 
\beq \label{e:reducedtensorproductcomposeinnerand}
B_{G_{Q^\rtensor}(1^{2+n})} = \left(\begin{matrix}
\alpha_1 \ket{t_1} & \beta_1 \ket{\gamma} & 0 & \alpha_1 C_1 \\
\alpha_2 \ket{t_1} & 0 & \beta_2 \ket{\gamma} & \alpha_2 C_1 \\
\alpha \ket{t_2} & 0 & 0 & \alpha C_2
\end{matrix}\right)
 \enspace .
\eeq
For comparison, the tensor-product-composed span program $Q^\otimes$ from~\cite[Def.~4.4]{Reichardt09spanprogram} would have the biadjacency matrix 
\beq
B_{G_{Q^\rtensor}(1^{2+n})} = \left(\begin{matrix}
\alpha_1 \ket{t_1} & \beta_1 \ket{\gamma} & 0 & \alpha_1 C_1 \\
\alpha_2 \ket{t_1} & 0 & \beta_2 \ket{\gamma} & \alpha_2 C_1 \\
\alpha_1 \ket{t_2} & 0 & 0 & \alpha_1 C_2 \\
\alpha_2 \ket{t_2} & 0 & 0 & \alpha_2 C_2
\end{matrix}\right)
 \enspace .
\eeq
\item
Next consider the case that $P' = P_{\OR}(s_1, s_2)$ from \defref{t:andorspanprogramdef}.  Let $Q^\rtensor$ be the composed span program from \defref{t:reducedtensorproductcomposedef}.  Then 
\beq \label{e:reducedtensorproductcomposeinneror}
B_{G_{Q^\rtensor}(1^{2+n})} = \left(\begin{matrix}
\delta \ket{t_1} & \epsilon_1 \ket{\gamma} & \epsilon_2 \ket{\gamma} & \delta C_1 \\
\delta \ket{t_2} & 0 & 0 & \delta C_2
\end{matrix}\right)
 \enspace .
\eeq
\end{itemize}

From Eqs.~\eqnref{e:reducedtensorproductcomposeinnerand} and~\eqnref{e:reducedtensorproductcomposeinneror}, we can derive a bound on the growth of the norm of the entrywise absolute value of the biadjacency matrix for a span program, under reduced tensor-product composition with either $P_{\AND}$ or $P_{\OR}$: 

\begin{lemma} \label{t:reducedtensorproductcomposeinnerandornorm}
Let $P$ be a strict, monotone span program on $1+n$ input bits, with $\abs{I_{1,1}} = 1$.  For $s_1, s_2 > 0$, let $P'$ be either $P_{\AND}(s_1, s_2)$ or $P_{\OR}(s_1, s_2)$, from \defref{t:andorspanprogramdef}.  Let $Q^\rtensor$ be the reduced tensor-product composition of $P'$ into the first input of $P$, as in \defref{t:reducedtensorproductcomposedef}.  Then 
\beq \label{e:reducedtensorproductcomposeinnerandornorm}
\norm{\abst(\biadj_{G_{Q^\rtensor}(1^{2+n})})}^2 \leq \frac{\sqrt{s_1} + \sqrt{s_2}}{\sqrt{s_1 + s_2}} \norm{\abst(\biadj_{G_P(1^{1+n})})}^2
 \enspace .
\eeq
\end{lemma}

\begin{proof}
Indeed, first consider the case that $P' = P_{\OR}(s_1, s_2)$.  By rearranging and regrouping the columns of the biadjacency matrices, from Eq.~\eqnref{e:reducedtensorproductcomposestart}, $\abst(B_{G_P(1^{1+n})})$ can be rewritten as 
\beq
\left(\begin{matrix}A & B\end{matrix}\right)
\eeq
for some entry-wise nonnegative matrices $A$ and $B$, such that from Eq.~\eqnref{e:reducedtensorproductcomposeinneror}, $\abst(B_{G_{Q^\rtensor}(1^{2+n})})$ can be rewritten as 
\beq
\left(\begin{matrix}
\delta A & \epsilon_1 B & \epsilon_2 B
\end{matrix}\right)
 \enspace .
\eeq
Now if $\delta$ were $\sqrt{\epsilon_1^2 + \epsilon_2^2}$, then the norm of this latter matrix would be exactly $\sqrt{\epsilon_1^2 + \epsilon_2^2}$ times the norm of the former matrix.  Since in fact $\delta = 1 < \sqrt{\epsilon_1^2 + \epsilon_2^2}$, we have the desired inequality~\eqnref{e:reducedtensorproductcomposeinnerandornorm}.  

Next consider the case that $P' = P_{\OR}(s_1, s_2)$.  Rearrange the columns in Eqs.~\eqnref{e:reducedtensorproductcomposestart} and~\eqnref{e:reducedtensorproductcomposeinnerand} to rewrite $\abst(B_{G_P(1^{1+n})})$ and $\abst(B_{G_{Q^\rtensor}(1^{2+n})})$ as, respectively, 
\beq
\left(\begin{matrix}
A & B \\
C & 0
\end{matrix}\right)
\qquad \text{and} \qquad
\left(\begin{matrix}
\alpha_1 A & \beta_1 B & 0 \\
\alpha_2 A & 0 & \beta_2 B \\
\alpha C & 0 & 0
\end{matrix}\right)
 \enspace ,
\eeq
for some entry-wise nonnegative matrices $A, B, C$.  Now if $\beta_1$ and $\beta_2$ were equal to $\alpha$, then the norm of the right biadjacency matrix would be exactly $\alpha$ times the norm of the left matrix.  Since in fact $\beta_1 = \beta_2 = 1 < \alpha$, Eq.~\eqnref{e:reducedtensorproductcomposeinnerandornorm} holds.  
\end{proof}

Let $\varphi$ be an AND-OR formula of size $n$, in which each AND and each OR gate has fan-in two.  Let $r$ be the root of $\varphi$.  For each vertex $v$ in $\varphi$, let $P_v$ be the span program 
\beq \label{e:andorPv}
P_v = \begin{cases} P_{\AND}\big(s_1(v), s_2(v)\big) & \text{if $g_v$ is an $\AND$ gate} \\ P_{\OR}\big(s_1(v), s_2(v)\big) & \text{if $g_v$ is an $\OR$ gate} \end{cases}
\eeq
Denote by $\alpha(v), \alpha_j(v), \beta_j(v)$ and $\epsilon(v), \delta(v), \epsilon_j(v)$, for $j \in [2]$, the parameters of $P_v$ from \defref{t:andorspanprogramdef}.  

Let $P_\varphi$ be the span program formed by composing the span programs $P_v$ according to the structure of $\varphi$ from the root $r$ of $\varphi$ toward the leaves, in an otherwise arbitrary order, using reduced tensor-product composition with respect to the bases of \defref{t:andorspanprogramdef}.  Note that $P_\varphi$ is strict and monotone.  If for each $v$, $s_1(v)$ and $s_2(v)$ are set to be the sizes of the respective input subformulas, then $\wsize {P_\varphi} = \sqrt n$.  \figref{f:reducedtensorproductcompositionexamples} has several examples.  

\begin{figure}
\centering
\begin{tabular}{c@{$\quad$}c@{$\quad$}c}
\subfigure[$x_1 \vee x_2$]{\label{f:orgate}\raisebox{.2in}{\includegraphics[scale=1]{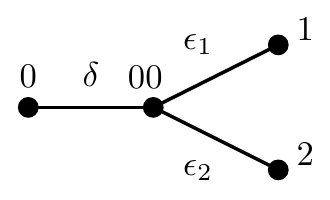}}}&
\subfigure[$x_1 \wedge x_2$]{\label{f:andgate}\includegraphics[scale=1]{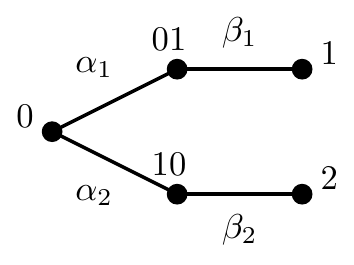}}\\
\subfigure[$(x_1 \wedge x_2) \vee x_3$]{\label{f:or_and}\raisebox{.05in}{\includegraphics[scale=1]{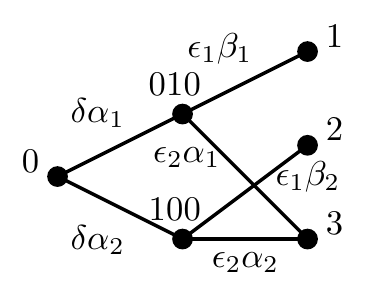}}}&
\subfigure[$(x_1 \vee x_2) \wedge x_3$]{\label{f:and_or}\includegraphics[scale=1]{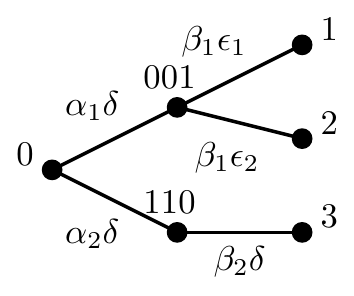}}
&\raisebox{.05cm}{\begin{picture}(140,0)\subfigure[$\big((x_1 \wedge x_2) \vee x_3\big) \wedge x_4$]{\label{f:and_or_and}\includegraphics[scale=1]{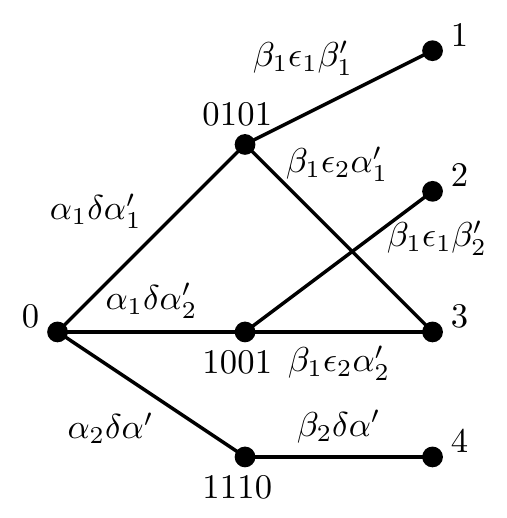}}\end{picture}}
\end{tabular}
\caption{Examples to illustrate reduced tensor-product composition for $\AND_2$-$\OR_2$ formulas.  For each formula $\varphi$, of size $n$, the graph $G_{P_\varphi}(1^n)$ is displayed.  Vertices corresponding to maximal false inputs are so labeled.  In (e), the primed variables refer to the span program coefficients for $x_1 \wedge x_2$.  Notice in each example that a vertex labeled with input $x \in \{0,1\}^n$ is connected exactly to those input bits $j$ with $x_j = 0$; this is a consequence of Eq.~\eqnref{e:reducedtensorproductANDORcomposeinput}.  Also notice that the graph's structure changes locally as each additional gate is composed onto the end of the formula, e.g., from (d) to~(e).  However, edge weights change nonlocally.  See also \figref{f:exampletensorproduct}.} \label{f:reducedtensorproductcompositionexamples}
\end{figure}

We can characterize $P_\varphi$ in terms of the set of ``maximal false" inputs, or minimal zero-certificates, to the formula $\varphi$.  

\begin{definition} \label{t:maximalfalsedef}
An input $x \in \B^n$ is a maximal false input to $\varphi$ if $\varphi(x) = 0$ and flipping any bit of $x$ from $0$ to $1$ changes the formula evaluation to $1$.  
\end{definition}

To any maximal false input $x$ corresponds a subtree $T_x$ of $\varphi$.  $T_x$ is rooted at $r$, and its leaves are exactly the leaves of $\varphi$ corresponding to input bits $k$ with $x_k = 0$.  Note that for each vertex $v \in T_x$, if $g_v$ is an $\OR$ gate, then both of $v$'s children are also in $T_x$; and if $g_v$ is an $\AND$ gate, then exactly one of $v$'s children is in $T_x$.  

\begin{lemma} \label{t:reducedtensorproductANDORcompose}
Let $U$ be the set of maximal false inputs to $\varphi$.  Then $P_\varphi$ is given as follows: 
\begin{itemize}
\item
Its inner product space is $V = \C^U$.  
\item
Its target vector is 
\beq \label{e:reducedtensorproductANDORcomposetarget}
\ket t = \sum_{x \in U} \left( \prod_{v \in T_x} \left\{ \begin{array}{cl} \alpha_{j(x, v)}(v) & \text{if $g_v = \AND_2$} \\ \delta(v) & \text{if $g_v = \OR_2$} \end{array}\right\} \cdot \prod_{v \notin T_x} \left\{ \begin{array}{cl} \alpha(v) & \text{if $g_v = \AND_2$} \\ \delta(v) & \text{if $g_v = \OR_2$} \end{array}\right\} \right) \ket x
 \enspace .
\eeq
Both products are over internal vertices only (the parameters $\alpha(v), \ldots$ are not defined for leaves).  In the first term $j(x, v) \in [2]$ indicates the child of $v$ that is in $T_x$.  
\item
Its input vectors are indexed by $I_{k1} = \{ k \}$, $I_{k0} = \emptyset$ for $k \in [n]$.  Letting $\gamma_k$ be the simple path from the $k$'th leaf to $r$, $\ket{v_k}$ is given by 
\beq \label{e:reducedtensorproductANDORcomposeinput}
\ket{v_k} = \sum_{x \in U : \, x_k = 0}
\left(
\prod_{v \in \gamma_k} \left\{ \begin{array}{c} \beta_{j(k, v)}(v) \\ \epsilon_{j(k, v)}(v) \end{array}\right\}
\cdot \prod_{v \in T_x \smallsetminus \gamma_k} \left\{ \begin{array}{c} \alpha_{j(x, v)}(v) \\ \delta(v) \end{array}\right\}
\cdot \prod_{v \notin T_x} \left\{ \begin{array}{c} \alpha(v) \\ \delta(v) \end{array}\right\}
\right)
\ket x
 \enspace .
\eeq
Here in each bracketed expression the top term is to be taken if $g_v$ is an AND gate, and the bottom term if $g_v$ is an OR gate.  The indices $j(x, v) \in [2]$ are as in the expression for $\ket t$, while $j(k, v) \in [2]$ indicates the child of $v$ that is along the path $\gamma_k$.  Products are over internal vertices only.  
\end{itemize}
In particular, for $x \in U$, $\braket{x}{v_k} = 0$ if and only if $x_k = 1$.  Thus $\ket x / \braket t x$ is a witness for $f_{P_\varphi}(x) = 0$.  
\end{lemma}

\begin{proof}
Although the expressions are intimidating, the proof is a simple induction.  Assume that we have completed reduced tensor-product composition of the span programs $P_v$ for vertices $v$ belonging to a subtree $\varphi'$ that includes $r$.  Then $\varphi'$ is also an AND-OR formula.  Assume that the characterization of \lemref{t:reducedtensorproductANDORcompose} holds for $\varphi'$.  

Consider adding a new vertex $v$ onto $\varphi'$, yielding a formula $\varphi''$.  We will use primed variables, $U'$, $T_x'$, $\gamma_k'$, to refer to $\varphi'$ and double-primed variables for $\varphi''$.  

The base case of the induction is if $\varphi'$ is empty and $\varphi'' = \{r\}$.  Then the claim is a consequence of \defref{t:andorspanprogramdef}.  If $g_r$ is an AND gate, then $U'' = \{ 01, 10 \}$, while $U'' = \{ 00 \}$ if $g_r$ is an OR gate.  

Now assume that the size of $\varphi'$ is at least two.  By symmetry, assume that the first input of $\varphi'$ is replaced by $v$.  Then by induction, note that $Z = \{ \text{maximal false inputs $x$ to $\varphi'$} \,\vert\, x_1 = 1 \}$ in \defref{t:reducedtensorproductcomposedef}.  There are two cases, depending on whether the gate $g_v$ is an OR or an AND.  
\begin{itemize}
\item
If the new gate is an OR, then the maximal false inputs of $\varphi'$ are in one-to-one correspondence to those of the new formula $\varphi''$.  By \defref{t:reducedtensorproductcomposedef} the inner product space does not change.  Moreover, the target vector is scaled simply by $\delta(v)$.  All but the first input vector of $P_{\varphi'}$ are scaled by $\delta(v)$.  The first input vector is split into two vectors, scaled by $\epsilon_1(v)$ and $\epsilon_2(v)$.  By examination of Eqs.~\eqnref{e:reducedtensorproductANDORcomposetarget} and~\eqnref{e:reducedtensorproductANDORcomposeinput}, the induction assumption is maintained.  
\item
If the new gate is an AND, then a maximal false input $1x$ of $\varphi'$ corresponds to the maximal false input $11x$ of $\varphi''$, while any maximal false input $0x$ of $\varphi'$ splits into the two maximal false inputs $01x$ and $10x$ of $\varphi''$.  Consider first a maximal false input $1x$.  Since $1x \in Z$, the target and all input vectors are simply scaled by $\alpha(v)$ in this coordinate.  This satisfies the induction hypothesis since $v \notin T_{11x}''$.  

Next consider a maximal false input $0x$ to $\varphi'$; $v \in T_{01x}'' \cap T_{10x}''$.  By Eq.~\eqnref{e:reducedtensorproductcomposedeftarget}, the $0x$ coordinate of the target vector splits in two, weighted by $\alpha_1(v)$ and $\alpha_2(v)$, so Eq.~\eqnref{e:reducedtensorproductANDORcomposetarget} continues to hold.  
Similarly, for the unchanged inputs $k$, i.e., inputs with $v \notin \gamma_k''$, $\Big(\begin{smallmatrix} \braket{01x}{v_k''} \\ \braket{10x}{v_k''} \end{smallmatrix}\Big) = \Big(\begin{smallmatrix} \alpha_1(v) \\ \alpha_2(v) \end{smallmatrix}\Big) \braket{0x}{v_k'}$ by Eq.~\eqnref{e:reducedtensorproductcomposedefinput}.  This is accounted for by the second term of Eq.~\eqnref{e:reducedtensorproductANDORcomposeinput}.  By Eq.~\eqnref{e:reducedtensorproductcomposedefinput}, $\Big(\begin{smallmatrix} \braket{01x}{v_1''} \\ \braket{10x}{v_1''} \end{smallmatrix}\Big) = \Big(\begin{smallmatrix} \beta_1(v) \\ 0 \end{smallmatrix}\Big) \braket{0x}{v_1'}$.  This shows up in the first term of Eq.~\eqnref{e:reducedtensorproductANDORcomposeinput}, since $v \in \gamma_1''$.  A similar argument holds for the input vector $\ket{v_2''}$.  
\qedhere
\end{itemize}
\end{proof}

One straightforward consequence of \lemref{t:reducedtensorproductANDORcompose} is that the order in which gates are composed in $P_\varphi$ does not matter.  In fact, the composition order does not matter for direct-sum, tensor-product or reduced tensor-product composition of arbitrary span programs.  Rather than prove this claim, though, in the sequel we will continue to order composition from the root toward the leaves.  

Since its construction uses reduced tensor-product composition, $P_\varphi$ has similar properties as a canonical span program~\cite[Def.~5.1]{Reichardt09spanprogram}.  A canonical span program has a dimension for every false input, whereas $P_\varphi$ has a dimension only for each maximal false input.  Unlike a canonical span program, $P_\varphi$'s target vector does not in general have uniform weights.  However, if $x \in \B^n$ has $\varphi(x) = 0$, let $x'$ be a maximal false input that lies above $x$, i.e., the bitwise AND of $x$ and $x'$ is $x$.  Then $\ket{x'} / \braket{t}{x'}$ is a witness for $f_{P_\varphi}(x) = 0$.  This property is a main motivation for defining reduced tensor-product span program composition; it does not hold for tensor-product composition. 

Unfortunately, a construction based only on reduced tensor-product span program composition will not work for applying \thmref{t:generalspanprogramalgorithmnonblackbox}.  There are two problems: 
\begin{enumerate}
\item
First, the number of maximal false inputs can be exponentially large in the size $n$ of the formula.  For example, if $\varphi$ is a balanced formula of depth $d$ with alternating levels of $\AND_2$ and $\OR_2$ gates, starting with AND, then there are $\frac{1}{2} 4^{2^{\lfloor(d - 1)/2\rfloor}} = e^{\Omega(\sqrt n)}$ maximal false inputs.  Since there are only $n$ input vectors, this implies that the maximum degree of a vertex in the graph must be exponentially large.  
\item
Second, the norm $\norm{\abst(A_{G_{P_\varphi}})}$ can also be exponentially large.  Indeed, $\alpha(v)$ can be as large as $2^{1/4}$ when the two input subformulas to $v$ have the same size.  Then, for example, a balanced formula of depth $d$ with only $\AND_2$ gates leads to a target vector with coefficients each $2^{\frac{1}{4} (2^d - (2d+1))}$.  
\end{enumerate}
The advantage of reduced tensor-product composition, though, is that because it outputs a \emph{strict} span program $P_\varphi$, the full witness size can be easily bounded in terms of the witness size.  

The disadvantages of reduced tensor-product composition are worst for very balanced formulas, while its advantage is most helpful for unbalanced formulas.  This suggests that a combining the two techniques might be useful for general general AND-OR formulas.  

Before turning to such an approach, though, let us first restate \lemref{t:reducedtensorproductANDORcompose} for the case of a maximally unbalanced formula.  \lemref{t:reducedtensorproductANDORcomposemaxunbalanced} is a key structural claim behind our proof of \thmref{t:andor}.  Once again, \figref{f:reducedtensorproductcompositionexamples} has several examples.  

\begin{lemma} \label{t:reducedtensorproductANDORcomposemaxunbalanced}
Let $\psi$ be a maximally unbalanced AND-OR formula, i.e., the depth of the formula equals the number of gates $J$.  Index the inputs in order from farthest to closest to the root (see \figref{f:gatepath}).  Let $T_{\AND} = \{ j \in [J+1] : \text{$j$ is an input to an AND gate} \}$, and $T_{\OR} = [J+1] \smallsetminus T_{\AND}$.  

Then the maximal false inputs $U$ to $\psi$ are in one-to-one correspondence to the elements of $T_F = T_{\AND} \cup \{ 1 \}$.  To $j \in T_F$ corresponds the maximal false input $x^j \in \B^{J+1}$ given by 
\beq
x^j_k = \begin{cases}
0 & \text{if $k = j$, or if $k > j$ and $k \in T_{\OR}$} \\
1 & \text{if $k < j$, or if $k > j$ and $k \in T_{\AND}$} 
\end{cases}
\eeq

The reduced-tensor-product-composed span program $P_\psi$ is given as follows: 
\begin{itemize}
\item
Its inner product space is $V = \C^U$.  
\item
Its target vector is 
\beq \label{e:reducedtensorproductANDORcomposemaxunbalancedtarget}
\ket t = \sum_{j \in T_F} \left( \prod_{v \in \gamma_j} \left\{ \begin{array}{c} \alpha_{\iota(j, v)}(v) \\ \delta(v) \end{array}\right\} \cdot \prod_{v \notin \gamma_j} \left\{ \begin{array}{c} \alpha(v) \\ \delta(v) \end{array}\right\} \right) \ket{x^j}
 \enspace .
\eeq
Here, in the first term $\iota(j, v) \in [2]$ indicates the child of $v$ that is in $\gamma_j$, the simple path from the $j$th vertex to the root.  In each bracketed expression the top term is to be taken if $g_v$ is an AND gate, and the bottom term if $g_v$ is an OR gate.  The products are over internal vertices~$v$.  
\item
Its input vectors are indexed by $I_{k1} = \{ k \}$, $I_{k0} = \emptyset$ for $k \in [J+1]$.  For $k \in T_{\AND} \cup [\min T_{\AND}]$, $\ket{v_k}$ is given by 
\beq \label{e:reducedtensorproductANDORcomposemaxunbalancedinputAND}
\ket{v_k} = 
\prod_{v \in \gamma_k} \left\{ \begin{array}{c} \beta_{\iota(k, v)}(v) \\ \epsilon_{\iota(k, v)}(v) \end{array}\right\}
\cdot \prod_{v \notin \gamma_k} \left\{ \begin{array}{c} \alpha(v) \\ \delta(v) \end{array}\right\}
\ket{x^\kappa}
\eeq
where $\kappa = k$ if $k \in T_{\AND}$ and $\kappa = 1$ otherwise.  For $k \notin T_{\AND} \cup [\min T_{\AND}]$, 
\beq \label{e:reducedtensorproductANDORcomposemaxunbalancedinputOR}
\ket{v_k} = \sum_{\substack{j \in T_F \\ j < k}}
\left(
\prod_{v \in \gamma_k} \left\{ \begin{array}{c} \beta_{\iota(k, v)}(v) \\ \epsilon_{\iota(k, v)}(v) \end{array}\right\}
\cdot \prod_{v \in \gamma_j \smallsetminus \gamma_k} \left\{ \begin{array}{c} \alpha_{\iota(j, v)}(v) \\ \delta(v) \end{array}\right\}
\cdot \prod_{v \notin \gamma_j} \left\{ \begin{array}{c} \alpha(v) \\ \delta(v) \end{array}\right\}
\right)
\ket{x^j}
 \enspace .
\eeq
The same conventions are understood as in the expression for $\ket t$.  
\end{itemize}
In particular, for $x \in U$, $\ket x / \braket t x$ is a witness for $f_{P_\psi}(x) = 0$.  
\end{lemma}

\begin{proof}
The lemma follows from \lemref{t:reducedtensorproductANDORcompose}.  The main simplification to make is that the tree $T_{x^j} = \gamma_j \cup \{ k \in T_{\OR} : k > j \}$, so $T_{x^j}$ and $\gamma_j$ agree on internal vertices.  Then Eq.~\eqnref{e:reducedtensorproductANDORcomposeinput} simplifies~to 
\beq
\ket{v_k} = \sum_{j \in T_F : \, x^j_k = 0}
\left(
\prod_{v \in \gamma_k} \left\{ \begin{array}{c} \beta_{\iota(k, v)}(v) \\ \epsilon_{\iota(k, v)}(v) \end{array}\right\}
\cdot \prod_{v \in \gamma_j \smallsetminus \gamma_k} \left\{ \begin{array}{c} \alpha_{\iota(j, v)}(v) \\ \delta(v) \end{array}\right\}
\cdot \prod_{v \notin \gamma_j} \left\{ \begin{array}{c} \alpha(v) \\ \delta(v) \end{array}\right\}
\right)
\ket{x^j}
 \enspace .
\eeq
This further simplifies to Eq.~\eqnref{e:reducedtensorproductANDORcomposemaxunbalancedinputAND} when $k \in T_{\AND} \cup [\min T_{\AND}]$ because then the sum over $j \in T_F$ with $x^j_k = 0$ has only a single term, either $k$ itself or $1$ if $k < \min T_{\AND}$.  It simplifies to Eq.~\eqnref{e:reducedtensorproductANDORcomposemaxunbalancedinputOR} when $k \notin T_{\AND} \cup [\min T_{\AND}]$ because then $x^j_k = 0$ for $j \in T_F$ exactly when $j < k$.  
\end{proof}

\subsection{AND-OR span program construction using hybrid composition} \label{s:hybridconstruction}

Let $\varphi$ be an AND-OR formula in which every gate has fan-in two, possibly after expanding higher fan-in gates.  We now construct the span program $P_\varphi$ on which the algorithm in \thmref{t:andor} is based.  For each internal vertex $v$ of $\varphi$, fix the parameters $s_1(v)$ and $s_2(v)$ to the sizes of the respective input subformulas.  Assume without loss of generality that $s_1(v) \geq s_2(v)$ always.  

To begin the construction, we mark certain edges of the formula.  Marked edges are termed ``checkpoints," because they will serve to cut off reduced tensor-product span program compositions and the associated exponential growths in norm and degree.  There are two steps to placing the checkpoints: 
\begin{enumerate}
\item
For every internal vertex $v$, i.e., a vertex that is not a leaf, mark the edge to the smaller of its two input subformulas.  Break ties arbitrarily.  

After this step, every gate has exactly one unmarked input edge, so the marked edges divide $\varphi$ into a set of paths.  Let 
\beq \label{e:smallpaths}
\cS = \{ \text{vertices $v$} : \text{for one of the paths, $v$ is the endpoint closer to the root $r$} \}
 \enspace .
\eeq
Certainly the root $r$ is itself in $\cS$.  

\item
In the second step, place more checkpoints to split up paths that are too long.  For each path, apply the following rule: 

Starting at the far end of the path and moving toward the root, keep track of the product 
\beq \label{e:checkpointrule}
\prod_v \left\{\begin{array}{cl} \alpha(v) & \text{if $g_v = \AND_2$} \\ \epsilon(v) & \text{if $g_v = \OR_2$} \end{array}\right\}
= 
\prod_v \bigg( \frac{\sqrt{s_1(v)}+\sqrt{s_2(v)}}{\sqrt{s_1(v) + s_2(v)}} \bigg)^{1/2}
\eeq
for internal vertices $v$ along the path.  Note that $\alpha(v), \epsilon(v) \leq 2^{1/4} \approx 1.19$.  After adding a vertex that makes the product exceed $\sqrt e \approx 1.65$, split the path by adding a checkpoint.  

After finishing this step, there will be no paths with the above product more than $2^{1/4} \sqrt e$, and for all paths except possibly those ending at a vertex in $\cS$, the product will be at least~$\sqrt e$.  
\end{enumerate}
We remark that our analysis in \secref{s:andoranalysis} is not overly sensitive to modifying these rules, for example by using a different constant greater than one instead of $\sqrt e$ in the second step.  

\bigskip

Based on the checkpointed formula, we can now construct the composed span program $P_\varphi$: 

\begin{enumerate}
\item
First, for each path $\xi$, compose the span programs $P_v$ from Eq.~\eqnref{e:andorPv} along the path, using reduced tensor-product composition, to obtain a span program $P_\xi$.  
\item
Next, apply direct-sum composition to compose the $P_\xi$ span programs constructed in the first step across the checkpointed edges.  
\end{enumerate}

Let $P_\varphi$ be the resulting span program.  \figref{f:exampletensorproduct} shows an example.  
Notice that when reduced tensor-product and direct-sum span program composition are both used, 
their relative order matters.  In constructing $P_\varphi$, all direct-sum composition comes last.

\subsection{Analysis of $P_\varphi$} \label{s:andoranalysis}

\begin{proof}[Proof of \thmref{t:andor}]
From \claimref{t:andorspanprogramwsize} and the span program composition results~\cite[Theorem~4.3, Prop.~4.7]{Reichardt09spanprogram}, we obtain: 

\begin{lemma}
$P_\varphi$ computes $f_{P_\varphi} = \varphi$ and $\wsize {P_\varphi} = \sqrt n$.  
\end{lemma}

In order to apply \thmref{t:generalspanprogramalgorithmnonblackbox}, the basic idea is to treat the gates along each checkpointed path $\xi$ as a single grouped gate, and to analyze the direct-sum composition of these grouped gates as in the proof of~\cite[Theorem~1.11]{Reichardt09unbalancedformula}.  There are two steps to this analysis.  
\begin{enumerate}
\item
First, we argue that $\norm{\abst(A_{G_{P_\varphi}})} = O(1)$ and $G_{P_\varphi}$ has maximum degree~$O(\sqrt n)$.  
\item
Second, to bound the full witness size we apply \lemref{t:wsizefcompose} to the direct-sum composition of the grouped gates, similarly to~\cite[Lemmas~3.3, 4.4]{Reichardt09unbalancedformula}.  
\end{enumerate}

Let us establish a bound on $\norm{A_{G_{P_\varphi}}}$.  Since $P_\varphi$ has nonnegative entries, $\abst(A_{G_{P_\varphi}}) = A_{G_{P_\varphi}}$.  

\begin{lemma} \label{t:andorunbalancednorm}
The norm of the adjacency matrix $A_{G_{P_\varphi}}$ is at most $2 (2 \sqrt 2 e + 1) = O(1)$.
\end{lemma}

\begin{proof}
By \lemref{t:directsumnorm}, it suffices to bound the norm of $A_{G_{P_\xi}}$ for any checkpointed path $\xi$.  

By definition, if a path $\xi$ involves $J$ internal vertices, 
\beq
A_{G_{P_\xi}} = 
\left(\begin{matrix}
0 & 0 & B \\
0 & 0 & C \\
B^\dagger & C^\dagger & 0
\end{matrix}\right)
\eeq
where $B$ is the biadjacency matrix $\biadj_{G_{P_\xi}(1^{J+1})}$ and $C$ is a column of zeros followed by an identity matrix.  Therefore we bound
\beq
\norm{A_{G_{P_\xi}}} \leq \left\lVert\left(\begin{matrix} 0 & B \\ B^\dagger & 0 \end{matrix}\right)\right\rVert +\left\lVert\left(\begin{matrix} 0 & C \\ C^\dagger & 0 \end{matrix}\right)\right\rVert = \norm{B}^2 + 1
 \enspace .
\eeq
Now from \lemref{t:reducedtensorproductcomposeinnerandornorm} and the checkpoint rule Eq.~\eqnref{e:checkpointrule}, $\norm{B}^2 \leq \sqrt 2 e \norm{(\begin{matrix}1 & 1\end{matrix})}{}^2 = 2 \sqrt 2 e$.  
\end{proof}

For a path $\xi$, let $\psi(\xi)$ be the AND-OR formula that is the composition of the gates along the path $\xi$.  $\psi(\xi)$ is maximally unbalanced, so \lemref{t:reducedtensorproductANDORcomposemaxunbalanced} applies to $P_\xi = P_{\psi(\xi)}$.  In particular, this allows us to bound the maximum degree of a vertex in $G_{P_\xi}$, and to bound the lengths of witnesses to $f_{P_{\psi(\xi)}}(x) = 0$: 

\begin{corollary} \label{t:andorunbalanceddegree}
The maximum degree of a vertex in $G_{P_\varphi}$ is $O(\sqrt n)$.  
\end{corollary}

\begin{proof}
Since $P_\varphi$ is constructed by direct-sum composition of span programs $P_\xi$, the maximum degree of a vertex in $G_{P_\varphi}$ is at most twice the maximum degree of a vertex in a graph $G_{P_\xi}$, which is at most twice the number of vertices in $G_{P_\xi}$.  

Fix a path $\xi$ with $J$ gates.  By \lemref{t:reducedtensorproductANDORcomposemaxunbalanced}, the number of vertices in $G_{P_\xi}(1^{J+1})$ is $1 + \abs{T_F} + (J+1) \leq 2J + 3$.  Now $\prod_{v \in \xi} \alpha(v) \leq 2^{1/4} \sqrt e$, but each $\alpha(v)$ satisfies 
\beq
\alpha(v)^2 > \frac{1 + \frac{1}{\sqrt n}}{\sqrt{1 + \frac 1 n}} > 1 + \frac{1}{2 \sqrt n}
 \enspace ,
\eeq
where in the first inequality we used $s_2(v) \geq 1$ and $s_1(v) < n$.  Hence $J = O(\sqrt n)$.  
\end{proof}

\begin{corollary} \label{t:andorunbalancedwitnessfalse}
For a checkpointed path $\xi$, let $\psi = \psi(\xi)$, and recall from \lemref{t:reducedtensorproductANDORcomposemaxunbalanced} the definitions of $T_F = T_{\AND} \cup \{1\}$ and of the maximal false inputs $x^j$ for $j \in T_F$.  Let $m$ be the size of the subformula of $\varphi$ rooted at the root of $\xi$.  

Then for any $j \in T_F$, $\ket{w'_j} = \ket{x^j} / \braket{t}{x^j}$ is a witness for $f_{P_\xi}(x^j) = 0$, with 
\beq
\norm{\ket{w'_j}}{}^2 \leq \begin{cases} \sqrt{2} e & \text{if $j = 1$} \\ \sqrt{2} e \sqrt m & \text{otherwise} \end{cases}
\eeq
\end{corollary}

\begin{proof}
The coefficient $\braket t {x^j}$ has been computed in Eq.~\eqnref{e:reducedtensorproductANDORcomposemaxunbalancedtarget}.  Since $\alpha(v) \geq 1$ and $\delta(v) = 1$ always, we need to place an upper bound on $1 / \abs{\braket{t}{x^j}}^2 \leq \prod_{v \in \gamma_j} \alpha_{\iota(j, v)}(v)^{-2}$.  

Note that 
\beq
\frac{1}{\alpha_1(v)^2} = \sqrt{\frac{s_1(v) + s_2(v)}{s_1(v)}} = \sqrt{1 + r(v)}
 \enspace ,
\eeq
where $r(v) = s_2(v) / s_1(v) \leq 1$.  Since $\alpha(v)^2 = \big(1 + \sqrt{r(v)}\big) / \sqrt{1 + r(v)}$, therefore $\frac{1}{\alpha_1(v)^2} \leq \alpha(v)^2$.  On the other hand, the best bound we can place on $\frac{1}{\alpha_2(v)^2}$ is, since $s_2(v) \geq 1$, 
\beq
\frac{1}{\alpha_2(v)^2} = \sqrt{\frac{s_1(v) + s_2(v)}{s_2(v)}} \leq \sqrt{s_1(v) + s_2(v)}
 \enspace .
\eeq

When $j = 1$, $\iota(j, v) = 1$ for every $v$, so we obtain the claimed bound 
\beq
\norm{\ket{w'_j}}^2 \leq \prod_{v \in \xi} \alpha(v)^2 \leq \sqrt 2 e
 \enspace .
\eeq
For $j > 1$, all but one of the indices $\iota(j, v)$ will be $1$, and for the last internal vertex $v$ on $\gamma_j$, $\iota(j, v)$ will be $2$.  It follows that $\norm{\ket{w'_j}}^2 \leq (\prod_{v \in \xi} \alpha(v)^2) \sqrt m \leq \sqrt 2 e \sqrt m$.  
\end{proof}

Let $\varphi'$ be a formula whose internal vertices are the checkpointed paths in $\varphi$ and whose leaves correspond to the inputs of $\varphi$.  The formula $\varphi'$ is \emph{not} an AND-OR formula; its gate at an internal vertex $v'$ is the composition of the gates in the corresponding path $\xi(v')$ in $\varphi$, $g'_{v'} = f_{P_{\xi(v')}}$.  However, $\varphi'$ and $\varphi$ compute the same boolean function.  Let $r'$ be the root of $\varphi'$, corresponding to the path in $\varphi$ that includes $r$.  For each vertex $v'$ of $\varphi'$, let $\varphi'_{v'}$ be the subformula rooted at $v'$, and let $P_{\varphi'_{v'}}$ be the span program obtained by direct-sum composition of the span programs $P_{\xi(w')}$ for all vertices $w'$ in $\varphi'_{v'}$.  Then $P_\varphi = P_{\varphi'_{r'}}$.  Let $s_{v'}$ be the number of inputs to $\varphi'_{v'}$, so $\wsize {P_{\varphi'_{v'}}} = \sqrt{s_{v'}}$.  

Call an internal vertex $v'$ of $\varphi'$ \emph{small} if the corresponding path $\xi(v')$ has an endpoint in the set $\cS$ from Eq.~\eqnref{e:smallpaths}.  Aside from the root $r'$, every small vertex $v'$ has subformula size $s_{v'}$ no larger than the size of its sibling subformula in $\varphi$; hence the name.  

\begin{lemma} \label{t:andorsizegrowth}
Let $v'$ be an internal vertex of $\varphi'$, with children $c'_1, c'_2, \ldots, c'_{J+1}$, sorted in decreasing order of 
their distances from $r'$, and with $s_{c'_1} \geq s_{c'_2}$ (see \figref{f:gatepath}).  Then for $j \geq 2$, $s_{c'_j} \leq s_{v'}/2$.  If $v'$ is not a small vertex, then $s_{c'_1} \leq s_{v'} - \frac{1}{\sqrt 2} \sqrt{s_{v'}}$.  
\end{lemma}

\begin{figure}
\centering
\includegraphics[scale=1]{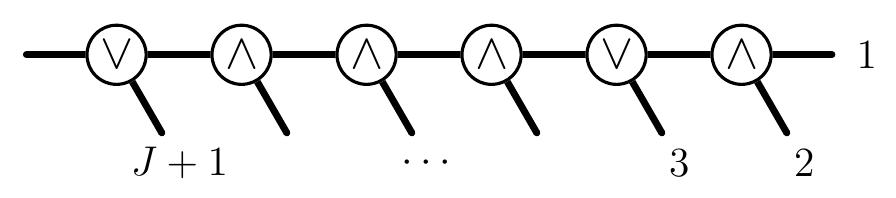}
\caption{A typical checkpointed path $\xi$ in $\varphi$ on $J = 6$ vertices.  The root $r$ is to the left.  Input subformulas are indexed from right to left; of $1$ and $2$, the larger subformula comes first.} \label{f:gatepath}
\end{figure}

\begin{proof}
In placing checkpoints, at every vertex $v$ the edge to the smaller input subformula is marked.  Thus the size of that subformula can be at most half the size of $\varphi_v$.  However the larger subformula can have size up to $s_v - 1$.  Assuming that $v'$ is not small, we will argue that in the second step of placing checkpoints, paths are allowed to grow long enough toward the root that the size increases significantly.  

Say that in the second step of placing checkpoints we begin with a single vertex whose input subformulas have size $s$ and $r_1 s$, where $r_1 \in [1/s, 1]$.  The next vertex toward the root will have another input subformula whose size is, say, $r_2 (s + r_1 s)$, where $r_2 \in [1/(s+r_1 s), 1]$.  Similarly the $j$th vertex towards the root $r$ will have a smaller input subformula of size $r_j s \prod_{k < j} (1 + r_k)$, so its two subformulas have total size $s \prod_{k \leq j} (1 + r_k)$.  Here, $r_j \leq 1$ and $1/r_j \leq s \prod_{k < j} (1 + r_k)$.  

Now it is possible that this path gets all the way to $\cS$ without a new checkpoint being placed.  Then this path corresponds to a small vertex, and we can make no strong claim about how much the size increases along the path.  For example, if the path \emph{starts} at the root $r$, then no further checkpoints can be placed, so we can only bound $s_{c'_1} \leq n-1$.  

Assume that a checkpoint is placed to terminate this path.  Then we have $\prod_j \frac{1 + \sqrt{r_j}}{\sqrt{1+r_j}} \geq e$.  We want to lower bound $\prod_j (1 + r_j)$.  Letting $x_j = \log \frac{1 + \sqrt{r_j}}{\sqrt{1+r_j}} \in (0, \frac12 \log 2]$, $\log (1 + r_j) = \log \frac{2}{1 + e^{x_j} \sqrt{2 - e^{2 x_j}}}$.  Applying Jensen's inequality gives $\prod_j (1 + r_j) \geq e^{1/J} > 1 + 1/J$, where $J$ is the number of vertices along the path.  

We have $s_{v'} = s_{c'_1} \prod_j (1 + r_j) \geq 2$.  If $J \geq \frac{1}{\sqrt 2} \sqrt{s_{v'}}$, then since at each step toward the root the size must increase by at least one, $s_{c'_1} \leq s_{v'} - J \leq s_{v'} - \frac{1}{\sqrt 2} \sqrt{s_{v'}}$, as claimed.  If $J < \frac{1}{\sqrt 2} \sqrt{s_{v'}}$, then using $s_{c'_1} < s_{v'} / (1 + 1/J)$ again gives $s_{c'_1} < s_{v'} - \frac{1}{\sqrt 2} \sqrt{s_{v'}}$.  
\end{proof}

In the maximally unbalanced formula, $s_{c'_1} = s_{v'} - O(\sqrt{s_{v'}})$, so aside from improving the constant $\frac{1}{\sqrt 2}$, the bound on $s_{c'_1}$ in \lemref{t:andorsizegrowth} is tight.  

\begin{lemma} \label{t:ANDORformulaevaluationtruefalsewitness}
Let $v'$ be an internal vertex of $\varphi'$, and let $m = s_{v'}$.  Then 
\beq \label{e:ANDORformulaevaluationtruefalsewitness}
\frac{ \wsizef {P_{\varphi'_{v'}}}}{ \sqrt m } \leq
\begin{cases}
\kappa \log m & \text{if $v'$ is not small} \\
\kappa \log m + \frac{\lambda}{\sqrt m} & \text{if $v'$ is small}
\end{cases}
\eeq
where $\lambda = \sqrt{2} e \approx 3.84$ and $\kappa = (1 + \frac 1 {\sqrt 2}) \lambda / \log 2 \approx 9.47$.  
\end{lemma}

\begin{proof}
The proof is by induction in the maximum distance from $v'$ to a leaf in $\varphi'$.  The base case is if all of $v'$'s inputs are leaves.  This case can be handled simultaneously to the induction step by letting $S = \emptyset$, so $\sigma = 1$, in the applications of \lemref{t:wsizefcompose} below.  

Take $v'$ an internal vertex, with children $c'_1, c'_2, \ldots, c'_{J+1}$, sorted in decreasing order of 
their distances from $r'$, and with $s_{c'_1} \geq s_{c'_2}$.  Note that each vertex $c'_j$ with $j \geq 2$ must be either small or a leaf.  However, $c'_1$ is not a small vertex.  Let $\xi = \xi(v')$ be the corresponding checkpointed path in $\varphi$.  Consider an input $x \in \B^m$.  We will use \lemref{t:wsizefcompose} in combination with \lemref{t:andorsizegrowth} and \corref{t:andorunbalancedwitnessfalse}.  

If $\varphi'_{v'}(x) = 1$, then apply \lemref{t:wsizefcompose} with $S = \{ j \in [J+1] : \text{$c'_j$ is not a leaf} \}$.  Since $P_{\xi(v')}$ is a strict span program, Eq.~\eqnref{e:wsizefcomposetrue} gives 
\beq
\frac{ \wsizefx {P_{\varphi'_{v'}}} x }{ \sqrt m } \leq \sigma + \frac{1}{\sqrt m}
 \enspace ,
\eeq
where $\sigma = \max_{j \in S} \wsizef {P_{\varphi'_{c'_j}}} / \sqrt{s_{c'_j}}$.  Now by \lemref{t:andorsizegrowth}, for $j \in S$ with $j \geq 2$, $s_{c'_j} \leq m/2$ and by induction 
\beq
\wsizef {P_{\varphi'_{c'_j}}} / \sqrt{s_{c'_j}} \leq \kappa \log {s_{c'_j}} + \lambda / \sqrt{s_{c'_j}}
 \enspace .
\eeq
Since $\kappa > \lambda / (2 \sqrt 2)$, the right-hand side is an increasing function of $s_{c'_j}$, so may be bounded by $\kappa \log(m/2) + \lambda / \sqrt{m/2}$.  For the case $j = 1$, by induction $\wsizef {P_{\varphi'_{c'_1}}} / \sqrt{s_{c'_1}} \leq \kappa \log {s_{c'_1}}$.  

If $v'$ is not small, then \lemref{t:andorsizegrowth} gives $s_{c'_1} \leq m - \sqrt{m/2}$, and therefore 
\beq
\sigma \leq \max \Big\{ \kappa \log (m - \sqrt{m/2}) , \kappa \log(m/2) + \frac{\lambda}{\sqrt{m/2}} \Big\}
 \enspace .
\eeq
The first term, bounded by $\kappa (\log m - 1 / \sqrt{2 m})$, is at most $\kappa \log m - 1 / \sqrt m$, provided that $\kappa \geq \sqrt 2$.  The second term satisfies the same bound provided $\kappa \geq (\lambda + \frac 1 {\sqrt 2}) / \log 2$.  

On the other hand, if $v'$ is small, then we can only be sure that $s_{c'_1} \leq s_{v'} - 1 = m - 1$.  This is all right, because Eq.~\eqnref{e:ANDORformulaevaluationtruefalsewitness} allows for more slack in this case.  In particular, it is enough to show that $\kappa \log (m/2) + \lambda / \sqrt{m/2} \leq \kappa \log m$, which indeed holds for $\kappa \geq \lambda/ \log 2$.  

Assume now that $\varphi'_{v'}(x) = 0$.  Let $x'$ be the input to $g'_{v'}$, i.e., the evaluations of $v'$'s input subformulas on $x$.  Let $y'$ be any minimal false input (\defref{t:maximalfalsedef}) to $g'_{v'}$ that lies above $x'$ (meaning that for every $k$, $y'_k \geq x'_k$).  In the notation of \lemref{t:reducedtensorproductANDORcomposemaxunbalanced}, $y' = x^j$ for some $j \in T_F$, and $\ket{w'_j} = \ket{x^j} / \braket t {x^j}$ is a witness for $f_{P_{\xi(v')}}(x') = 0$.  By Eq.~\eqnref{e:wsizefcomposefalse}, then, 
\beq\begin{split}
\frac{ \wsizefx {P_{\varphi'_{v'}}} x }{ \sqrt m }
&\leq 
\sigma(\bar x', \ket{w'_j}) + \frac{\norm{\ket{w'_j}}^2}{\sqrt m} \\
&\text{where $\sigma(\bar x', \ket{w'_j}) = \max_{\substack{k \in S : \\ \text{$x'_k = 0$ and $\braket{v_k}{w'_j} \neq 0$}}} \frac{\wsizef {P_{\varphi'_{c'_k}}}}{\sqrt{s_{c'_k}}}$}
 \enspace .
\end{split}\eeq
There are two cases to consider, $j > 1$ and $j = 1$.  

First, if $j > 1$, then by \lemref{t:reducedtensorproductANDORcomposemaxunbalanced}, $\braket{v_k}{w'_j} = 0$ for all $k < j$.  Therefore, $\sigma = \sigma(\bar x', \ket{w'_j})$ will only maximize over small vertices, $c'_k$ with $k > 1$.  By induction, $\sigma \leq \kappa \log (m/2) + \lambda / \sqrt{m/2}$.  On the other hand, \corref{t:andorunbalancedwitnessfalse} gives only $\norm{\ket{w'_j}}{}^2 \leq \lambda \sqrt m$.  Putting these bounds together, 
\beq\begin{split}
\frac{ \wsizefx {P_{\varphi'_{v'}}} x }{ \sqrt m }
&\leq
\kappa \log (m/2) + \lambda / \sqrt{m/2} + \lambda \\
&\leq
\kappa \log m
\end{split}\eeq
provided that $\kappa \geq 2 \lambda / \log 2$.  The above bound holds whether or not $v'$ is small.  

Next consider the case $j = 1$.  Then we will have a worse bound on $\sigma$, but \corref{t:andorunbalancedwitnessfalse} gives $\norm{\ket{w'_j}}^2 \leq \lambda$.  If $j = 1$ and $v'$ is not small, then \lemref{t:andorsizegrowth} gives $s_{c'_1} \leq m - \sqrt{m/2}$.  Using the induction hypothesis, we obtain 
\beq
\frac{ \wsizefx {P_{\varphi'_{v'}}} x }{ \sqrt m }
\leq
\max \Big\{ \kappa \log (m - \sqrt{m/2}), \kappa \log(m/2) + \frac{\lambda}{\sqrt{m/2}} \Big\} + \frac{\lambda}{\sqrt m}
 \enspace .
\eeq
This is indeed at most $\kappa \log m$ provided that $\kappa \geq (1 + \frac 1 {\sqrt 2}) \lambda / \log 2$.  

If $j = 1$ and $v'$ is small, then we can only be sure that $s_{c'_1} \leq m - 1$.  However, just as before, this is enough, since then $\kappa \log s_{c'_1} < \kappa \log m$, while $\kappa \log (m/2) + \lambda / \sqrt{m/2} \leq \kappa \log m$ for $\kappa \geq \lambda / \log 2$.  
\end{proof}

Substituting $v' = r'$ into \lemref{t:ANDORformulaevaluationtruefalsewitness} gives $\wsizef {P_{\varphi'_{r'}}} = \wsizef {P_\varphi} = O(\sqrt n \log n)$.  \lemref{t:andorunbalancednorm} bounds the norm of $A_{G_{P_\varphi}}$ and \corref{t:andorunbalanceddegree} bounds the maximum degree of a vertex.  Therefore, \thmref{t:generalspanprogramalgorithmnonblackbox} applies, completing the proof of \thmref{t:andor}.  
\end{proof}

\section{Open problems}

It is possible that our analysis here is loose.  \lemref{t:wsizefcompose} is conservative, as it uses the worst input ratio $\wsizef {P^j} / \wsize {P^j}$.  Potentially a more careful analysis of the full witness size across span programs $P_\xi$ could improve our upper bound.  However, we do not believe that the full witness size of $P_\varphi$ is $O(\sqrt n)$.  More ideas seem to be needed to find an $O(\sqrt n)$-query quantum algorithm for evaluating arbitrary size-$n$ AND-OR formulas, particularly if we wish the algorithm also to have poly-logarithmic time overhead after preprocessing.  One very simple idea would be to start with different span programs for AND and OR.  By~\cite[Lemma~4.12]{Reichardt09spanprogram}, there are many possibilities; the same invertible linear transformation can be applied to the target and all input vectors without changing the witness size.  We have unsuccessfully investigated choices to see whether reduced tensor-product composition alone would be sufficient to build an optimal span program for which $\norm{\abst(A_{G_{P_\varphi}})} = O(1)$, but have not studied other choices under hybrid span program composition.  

Some preprocessing will always be necessary if the formula $\varphi$ is presented poorly to the algorithm.  However, we would like the preprocessing to correspond to as natural a presentation of $\varphi$ as possible, and there is clearly more work to be done.  

\thmref{t:andor} and~\cite[Theorem~1.11]{Reichardt09unbalancedformula} can likely both be generalized to apply to formulas over a larger class of gates.  In particular, the idea of combining direct-sum and tensor-product span program composition might be useful in other contexts.  It leads to a dramatic speedup on AND-OR formulas that have not been rebalanced.  Since rebalancing and its effect on the adversary bounds is poorly understood for formulas over larger gate sets, this seems like promising technology to apply.  We have specialized to AND-OR formulas because they form an important class of formulas, and because the span programs for the AND and OR gates are especially simple.  Further applications will require studying other particular span programs more carefully.  

Although our algorithm can be analyzed without the span program formalism, it would not have been discovered without a certain familiarity with span programs.  Will span programs be useful for finding quantum algorithms for problems beyond formula evaluation?  One approach is to try to solve the semi-definite program (SDP) for the optimal span program witness size~\cite{Reichardt09spanprogram}.  This SDP is exponentially large, so presumably is difficult to solve exactly, but for special cases perhaps it can be solved to within a constant factor of the optimum.  Even if so, the span program this gives will typically correspond to a graph with large norm and high degree.  To turn this into an algorithm that is \emph{time} efficient, and not just query efficient, it will be useful to find techniques for breaking up large span programs, perhaps by using tailored semi-definite programs.  It is known how to convert a span program into a strict span program without affecting the witness size~\cite[Prop.~4.10]{Reichardt09spanprogram}.  However we do not know of any techniques for converting a span program with low witness size into a span program with low \emph{full} witness size that additionally has smaller norm and lower degree.

\section*{Acknowledgements}
I thank Andrew Childs and Robert {\v S}palek for helpful discussions.  Research supported by NSERC and ARO-DTO.

\bibliographystyle{alpha-eprint}
\bibliography{andor}

\end{document}